\documentclass[11pt,a4paper]{article}

\usepackage[utf8]{inputenc}
\usepackage[T1]{fontenc}
\usepackage[nice]{nicefrac}

\usepackage{fullpage} %
\usepackage[small]{caption}
\usepackage{algorithm}
\usepackage{algpseudocode}
\usepackage{amsmath}
\usepackage{amsfonts}
\usepackage{amssymb}
\usepackage{amsthm}
\usepackage{amstext}
\usepackage{thmtools}

\usepackage{csquotes}
\usepackage{xcolor}
\usepackage{graphicx}

\usepackage{cite} %
\usepackage{xspace}
\usepackage{booktabs, mathtools} 
\usepackage{adjustbox}
\usepackage{paralist} 
\usepackage{caption}
\usepackage{subcaption}
\usepackage[shortlabels]{enumitem}

\usepackage{color-edits}

\usepackage{url}
\usepackage[hidelinks]{hyperref}
\usepackage[capitalize]{cleveref}
\usepackage{amsthm,amsfonts,amsmath,amssymb}
\usepackage{thmtools}
\usepackage{dsfont}
\usepackage{geometry}
\usepackage{bbm,bm}
\usepackage{float}
\usepackage{color}
\usepackage{graphicx} 
\usepackage{fullpage}
\usepackage{todonotes}
\usepackage{bbm}

\setuptodonotes{backgroundcolor=blue!20, linecolor=blue!40}

\definecolor{toc}{RGB}{13,55,174}	
\definecolor{myPurple}{RGB}{175,0,124}

\usepackage[hidelinks]{hyperref}	
\hypersetup{
    colorlinks=true,
    citecolor=myPurple,
    linkcolor=toc,
    urlcolor=toc
}

\usepackage{enumitem}		
\usepackage{fontawesome5}   
\usepackage{adjustbox}
\allowdisplaybreaks			
\usepackage{subcaption}     
\usepackage{mathtools}
\usepackage{relsize}

\usepackage[symbol*]{footmisc}

\usepackage{thm-restate}
\newtheorem{theorem}{Theorem}[section]
\newtheorem{lemma}[theorem]{Lemma}
\newtheorem{corollary}[theorem]{Corollary}
 
\newtheorem{definition}[theorem]{Definition}

\renewcommand{\Pr}{\mathop{\bf Pr\/}}
\newcommand{\E}{\mathop{\bf E\/}}

\newcommand{\indic}{\mathbbm{1}}

\newcommand{\OPT}{\textnormal{OPT}}

\newcommand{\ALG}{\textnormal{ALG}}

\newcommand{\etal}{et al. }
\newcommand{\val}{\textnormal{value}}
\newcommand{\rank}{\textnormal{rank}}
\newcommand{\spann}{\textnormal{span}}

\newcommand{\reals}{\mathbb R}

\newcommand{\nats}{\mathbb N}

\newcommand{\calB}{\mathcal{B}}

\newcommand{\calD}{\mathcal{D}}

\newcommand{\calI}{\mathcal{I}}
\newcommand{\calJ}{\mathcal{J}}

\newcommand{\calO}{\mathcal{O}}
\newcommand{\calP}{\mathcal{P}}

\newcommand{\calS}{\mathcal{S}}

\newcommand{\calU}{\mathcal{U}}
\newcommand{\calV}{\mathcal{V}}

\newcommand{\calY}{\mathcal{Y}}
\newcommand{\calZ}{\mathcal{Z}}

\newcommand{\Select}{\text{Select}}

\newcommand{\Ind}[1]{\operatorname{\mathbb{I}}\left\{ #1 \right\}}

\newcommand{\bp}{\bm{p}}

\newcommand{\bx}{\bm{x}}

\def\<{\langle}
\def\>{\rangle}

\newcommand{\Be}{\textit{Be}}
\newcommand{\Bin}{\textit{Bin}}
\newcommand{\PB}{\textit{PB}}

\DeclareMathOperator*{\argmax}{argmax}
\DeclareMathOperator*{\argmin}{argmin}

\newcommand{\problemname}{\textsc{Portfolio Optimization }} 
\newcommand{\matroidproblemname}{\textsc{Matroid Portfolio Optimization }} 
\DeclarePairedDelimiter\abs{\lvert}{\rvert}%

\newcommand{\erc}{Supported by the Swiss State Secretariat for Education, Research and Innovation (SERI) under contract number MB22.00054.}

\title{Data-Driven Solution Portfolios}
\date{}

\author{
\makebox[3.5cm][c]{Marina Drygala} \\ EPFL \and
\makebox[3.5cm][c]{Silvio Lattanzi} \\ Google Research \and
\makebox[3.5cm][c]{Andreas Maggiori} \\ Columbia University \and \\
\makebox[3.5cm][c]{Miltiadis Stouras\thanks{\erc}} \\ EPFL \and \\
\makebox[3.5cm][c]{Ola Svensson\footnotemark[1]} \\ EPFL \and \\
\makebox[3.5cm][c]{Sergei Vassilvitskii} \\ Google Research
}

\begin{document}
\maketitle

\begin{abstract}

In this paper, we consider a new problem of portfolio optimization using stochastic information. 
In a setting where there is some uncertainty, we ask how to best select $k$ potential solutions, with the goal of optimizing the value of the best solution.  More formally, given a combinatorial problem $\Pi$, a set of value functions $\mathcal{V}$ over the solutions of $\Pi$, and a distribution $\mathcal{D}$ over $\mathcal{V}$, our goal is to select $k$ solutions of $\Pi$ that maximize or minimize the expected value of the {\em best} of those solutions.
For a simple example, consider the classic knapsack problem: given a universe of elements each with unit weight and a positive value, the task is to select $r$ elements maximizing the total value. Now suppose that each element's weight comes from a (known) distribution. How should we select $k$ different solutions so that one of them is likely to yield a high value?

In this work, we tackle this basic problem, and generalize it to the setting where the underlying set system forms a matroid. 
On the technical side, it is clear that the candidate solutions we select must be diverse and anti-correlated; however, it is not clear how to do so efficiently. 
Our main result is a polynomial-time algorithm that constructs a portfolio within a constant factor of the optimal.
\end{abstract}

\newpage \tableofcontents

\newpage
\section{Introduction}\label{sec: Introduction}

\noindent Worst-case analysis has long been the prevailing standard for assessing an algorithm's performance. However, this approach often falls short in capturing the real-world performance of algorithms, as instances encountered in practice often differ significantly from those that define worst-case scenarios. To address this discrepancy, the field of \emph{Data-Driven Algorithm Design} has sought to create algorithms that utilize past data about a problem, either implicitly or explicitly, and provably achieve superior performance on the typical instances that arise in practical applications. For surveys of this area see \cite{algorithms-with-predictions-survey, data-driven-survey}.

Building on this line of research, we study the portfolio optimization problem as a simple framework for speeding up algorithms using historical data; a method with potential applications across various combinatorial problems. Specifically, given a combinatorial problem $\Pi$, a solution set $\calI$, and a distribution $\calD$ over value functions (which map solutions to values, with each function representing a different scenario), our objective is to compute a portfolio of $k$ solutions that maximizes (or minimizes) the expected value (or cost) of the best among the $k$ solutions.
Starting from a different perspective, Kleinberg, Papadimitriou and Raghavan \cite{kleinberg2004segmentation}
have defined a very similar family of optimization problems, called ``Segmentation Problems''.
The \problemname problem can be seen as a stochastic version of the aforementioned family of problems; actually the two definitions are effectively equivalent when the distribution over value functions has a polynomial-sized support (for more details see Related Work, Section~\ref{sec:related-work}).

The \problemname problem captures many natural questions. The problem mentioned in the abstract--how to select $k$ bundles of $r$ items each so that the expected value of the best bundle is maximized--is only one of them.
Another practical application is finding the shortest path between two points (e.g., home and work) in a city, under varying daily traffic conditions. Here, the solution set consists of all possible source-destination paths and the value functions assign weights to each path based on specific traffic scenarios. The distribution over value functions models the stochastic nature of traffic. Rather than rerunning a shortest-path algorithm for each new instance, one could leverage the statistical knowledge about traffic patterns to precompute a few paths that cover different likely scenarios (e.g., peak morning traffic, light nighttime traffic, etc.). This approach, then, allows us to quickly evaluate the precomputed paths under new traffic conditions and select the best option without the need to resolve the problem from scratch each time.

Interestingly, this stochastic formulation also captures problems outside the area of speeding up algorithms. In particular, our objective captures any stochastic problem where one
needs to select a set of dependent objects with the goal of maximizing the expectation
of their maximum. A practical example of this is sports betting pools, which was recently explored in \cite{march-madness}.
The authors examine the scenario of participating in a betting pool for a basketball tournament, where individuals can pay a fee to submit a prediction for the outcome of all tournament matches, with the potential to win a substantial monetary prize if their predictions are accurate. Given a budget constraint on the number of entries that a person
can submit, the technical problem essentially reduces to computing a portfolio of entries
that maximizes the probability that one of the selected entries succeeds.

From a technical standpoint, in the maximization variant of our problem, the value
of a portfolio is a submodular function. Therefore, one can find an approximately optimal
portfolio by running the \textit{Greedy} algorithm, with a running time that is
polynomial in the size of the solution set $\calI$. However, for most interesting problems, e.g., spanning trees, entries in betting pools, etc., the solution set is given implicitly and is usually exponentially large in the size
of the input, making \textit{Greedy} inefficient. This raises the natural question of 
whether one can compute approximately optimal portfolios, for interesting
problems, with a running time that is polynomial in the size of the input.

\subsection{Our contributions}

The first important contribution of our paper is conceptual: we formulate the \problemname
problem from a new, stochastic, viewpoint and show that
it captures many interesting scenarios.

On the technical side, we focus on constructing solution portfolios for the fundamental
problem of optimization under matroid contraints. We examine the simplified case where
each element of the matroid's ground set takes a 0-1 value, independently, with some known 
probability, and the value of a set is the sum of the values of its elements. Despite its simplicity, this case captures many core challenges of constructing solution portfolios and proves to be technically challenging. As we discuss in Section~\ref{sec: Our Techniques}, this problem differs fundamentally from classical problems in the area of randomized algorithms, because constructing effective portfolios requires leveraging the “anti-concentration” properties of various solutions. Additionally, employing standard tools from the literature, such as contention resolution schemes, presents challenges due to the nature of our objective function. In particular, the value of a portfolio depends heavily on ``abnormal'' events, such as significant deviations from the expected value, occurring in one of its $k$ solutions. However, contention resolution schemes do not offer per-instance guarantees but rather only work on expectation. Therefore, we need to take extra care to apply these results while conditioning on those events.

The main technical contribution of our work is to design an algorithm with polynomial running time in the size of the matroid's ground set which constructs a portfolio that is a $\Theta(1)$-approximation of the optimal portfolio.

We give a high-level description of our techniques
in~\cref{sec: Our Techniques}, present a simpler algorithm for the case
of uniform matroids in~\cref{sec: Uniform Matroids} and present our general algorithm in~\cref{sec:matroids-algorithm}.

\subsection{Related work}\label{sec:related-work}

The related work can be broadly categorized into three main directions. The first is the study of Segmentation Problems, initially introduced by Kleinberg, Papadimitriou and Raghavan \cite{kleinberg2004segmentation}. This work, inspired by data mining techniques for market segmentation, defines a new class of optimization problems named \textit{Segmentation Problems}.
For any combinatorial optimization problem and a set $S$ of different cost vectors for this problem, the corresponding
segmentation problem asks to partition the set $S$ into several \textit{segments} and pick a separate solution for each segment, so that the total cost is minimized. The \problemname problem can be seen as a stochastic variant of the aforementioned family of problems, where instead of a fixed set of cost vectors, one has access to a distribution over cost vectors and aims to optimize the expected value of the constructed portfolio. In fact, the two formulations are effectively equivalent when the distribution over cost vectors has a polynomial-sized support. This new, stochastic, viewpoint can accommodate a wider variety of applications where randomness is inherent in the problem at hand (e.g., sports betting, as introduced earlier). Besides that, it also enables the formulation of elegant and technically challenging theoretical 
questions, like the case where every element of a groundset takes a value independently of the other elements.
Another adaptation of the family of ``Segmentation Problems'' is due to Gupta, Moondra and Singh \cite{robust-portfolios} who examined a "robust" version of the problem. In this variant, the objective is to identify a small set of solutions that guarantees a good approximation for each one of the cost vectors of interest. This approach was motivated by fairness concerns, aiming to guarantee a good approximation across various cost vectors that may arise due to different fairness constraints.

On the technical side, Kleinberg \etal \cite{kleinberg2004segmentation} study the Catalogue Segmentation problem, i.e. constructing portfolios for linear maximization with cardinality constraints.
For the case where elements take values in $\{0,1\}$, they design an algorithm that runs in time $O(n^{k \log k}/\delta)$ and produces a portfolio that is a $(1-\delta)$-approximation, under the assumption that on every instance a large fraction of the items have value $1$. Our result is a $O(1)$ approximation with a running
time that is polynomial in both $n$ and $k$, without the need of a density assumption, but for the case where
the value of each element is independent of the values of other elements. Furthermore, our results extend to more general settings, accommodating any matroid constraints on the ground set.

The second related line of work is the area of Data-Driven Algorithm Design and specifically speeding up an algorithms' execution by utilizing data.
Some of these attempts include speeding up: the \textit{Greedy} algorithm for submodular maximization~\cite{balkanski-data-summarization, krause-data-summarization}, the Hungarian algorithm for calculating maximum matchings~\cite{faster-matchings}, primal-dual algorithms for various graph problems \cite{graph-learned-primal-dual},
algorithms for flow problems \cite{ford-fulkerson, push-relabel} and the Bellman-Ford algorithm \cite{bellman-ford}. The key distinction of our work lies in defining a new set of problems and providing general techniques for solving them. 

A third relevant area of research is stochastic probing and online decision-making (for a detailed survey, see \cite{singla_thesis}). These problems typically involve a ground set of elements, each with an unknown stochastic weight sampled from a known distribution. The algorithm can probe certain elements—often incurring a cost— in order to reveal their weights. Based on these observations, it then selects a feasible subset and is rewarded with the weight of the selected items. A key distinction between this area and the \problemname problem is that, in our setting, the solutions are chosen entirely offline, without observing any weights, and remain fixed across all possible realizations.

Our problem is more closely aligned with a non-adaptive strategy for stochastic probing. However, the variants of non-adaptive strategies studied in the literature typically impose only partial restrictions on the algorithm’s flexibility—for instance, requiring the algorithm to preselect which elements to probe but still allowing it to form its solution after observing the probed weights. In contrast, in the \problemname problem, the selection is entirely offline, with no adjustments allowed after weights are realized. For this reasons, the benchmarks of the two problems are also different: non-adaptive strategies are evaluated against the optimal value achievable in hindsight, whereas our benchmark is the best offline strategy with the same constraints as the algorithm.

\section{Problem statement and preliminaries}\label{sec: Problem Statement and Preliminaries}

\subsection{Portfolio Optimization}
In this section we define the \problemname problem in its general form, for any classical combinatorial problem.
In the next section, we define the \matroidproblemname problem, which is the portfolio problem for a special case of 
maximization with matroid constraints.
\vspace{1em}

We present the maximization variant of our problem, with the minimization version defined in a similar manner.
For any combinatorial problem $\Pi$, the \problemname problem is described by a tuple 
$\calJ = (\calI, \calV, \calD, k)$ where $\calI$ is the set of feasible solutions
of problem $\Pi$, $\calV$ is a set of value functions from $\calI$ to $\reals_{\geq 0}$, $\calD$ is a 
distribution over $\calV$ and $k$ is a natural number that describes the desired size of the portfolio.

\vspace{1em}
Our goal is to select $k$ solutions from $\calI$ so as to maximize the expected value of the \textit{best} of those solutions. That is, for a collection of $k$ sets $\calS=\{S_1, \dots, S_k\}$ such that $S_i \in \calI, \forall i \in \{1, \dots, k\}$ we define its value as:
\begin{align*}
    \val(\calS) = \E_{v\sim\calD}\left[\max_{S_i \in \calS} v(S_i)\right].
\end{align*}

\noindent Formally, we want to solve the following optimization problem.
\begin{align*}
    \text{maximize}& \; \val (\calS) \\
    \text{s.t.}& \quad \calS = \{S_1, ..., S_k \} \text{ and } S_i \in \calI, \forall i \in [k].
\end{align*}

\noindent We remark that the solutions $S_1, \dots, S_k$ are chosen \textit{offline}, without observing the realized value function, but rather only by using our knowledge about the distribution $\calD$.

\vspace{1em}

For any input tuple $\calJ$, we denote by $\calO (\calJ) = \{O_1 (\calJ), \dots, O_k (\calJ) \}$ the optimal portfolio, that is the maximizer of the above optimization problem. To ease notation, whenever the input and the algorithm are clear from the context we use $\calO  = \{O_1 , \dots, O_k  \}$  and $\OPT$ to denote the optimum solution and its value respectively.

\vspace{1em}
For an algorithm $\ALG$ we denote by $\calS_{\ALG (\calJ)}$ its output on input $\calJ$.
We say that $\ALG$ is a $c$-approximation if it always outputs a collection of feasible solutions and:
\begin{align*}
     \E_{\ALG} \left[ \val \left(\calS_{\ALG (\calJ)} \right) \right] \geq c  \cdot \val \left( \calO (\calJ)  \right) ,\; \forall \calJ
\end{align*}
where the expectation is taken over the internal randomness of $\ALG$.

\vspace{1.5em}

In the maximization variant of the problem, the value of a set of solutions, defined as the expected maximum of their values, is a monotone 
submodular function, as demonstrated by Kleinberg and Raghu \cite{kleinberg2018team}. Consequently,
the celebrated \textit{Greedy} algorithm of \cite{nemhauser1978analysis} is a $(1-1/e)$ approximation
that runs in polynomial time in the size of the solution set $\calI$. Without imposing any further
restrictions to the problem, this is the best approximation ratio one can achieve, as the Max-$k$-Cover
problem can be reduced to the \problemname problem. The details of the reduction are deferred to Appendix~\ref{app:hard}.

Although Greedy achieves the optimal approximation ratio for this problem, its running time is impractical for 
most interesting applications like max-$k$-cover, knapsack, shortest paths, spanning trees, etc.
In all of the aforementioned applications, the set of feasible solutions is given implicitly
and is exponentially large in the size of the input. Therefore, one needs to have a running time that
is polynomial in the description of the solution set, rather than its size.
This raises the natural question of whether one can compute approximately optimal portfolios, 
for interesting problems, with a running time that is polynomial in the description of the solution set
$\calI$.

In this work, we focus on the fundamental case of optimization over matroid contraints and answer the latter question affirmatively for a natural distribution over value functions. 
We formulate this problem in the following section.

\subsection{Matroid Portfolio Optimization}

The \matroidproblemname problem is a special case of the \problemname problem, where the set of feasible solutions, $\calI$, is the family of independent sets of an underlying matroid $M = (E,\calI)$ (see~\cref{def: matroid}).
The set $\calV$ of value functions is the set of all \textit{additive} set functions that map the elements of $E$ to $\{0,1\}$. Formally, for any value function $u \in \calV$, element
$e \in E$ and subset $S \subseteq E$ we have: $u(\{e\}) \in \{0,1\}$ and $u(S) = \sum_{e' \in S} u(\{e'\})$.
Throughout the paper we use the term ``active'' to denote an element $e \in E$ that takes value one, under a specific value function, and the term ``inactive'' to describe an element with value zero.

We assume that $\calD$ is such that each element $e \in E$ is active with probability $p_e$, independently of the values of the rest of the elements of $E$. Formally, 
\begin{align*} 
\forall S &\subseteq E:
\Pr_{u \sim \calD} \left[ \bigwedge_{e \in S} \left\{ u(\{e\}) = 1 \right\}  \right] = \prod_{e \in S} p_e
\end{align*}

To simplify the notation we define an equivalent distribution $\calD'$ over subsets of $E$, that  represent the set of active elements, such that

$$
\forall e \in E: \Pr_{u \sim \calD}\left[ u(\{e\}) = 1 \right] = \Pr_{A \sim \calD'} \left[ e \in A\right].
$$
and

\begin{align*} 
\forall S &\subseteq E:
\Pr_{A \sim \calD'} \left[ \bigwedge_{e \in S} \left\{ e \in A \right\}  \right] = \prod_{e \in S} p_e
\end{align*}

\noindent In other words, instead of sampling a value function that assigns $\{0,1\}$ values to the elements, we can equivalently sample the set of active elements and, thus, switch to distributions over subsets of $E$. Also, the value of any set $S \subseteq E$ will now be equal to the random variable $|S \cap A|$, where $A$ is the random set denoting the active elements.

\vspace{2em}
For the remainder of the paper, the input to the \matroidproblemname problem will be described by
a triplet $\calJ = (M, k, \calD)$, where $M = (E, \calI)$ is a matroid over the ground set $E$,
$k$ a natural number describing the desired size of the portfolio and $\calD$ is a distribution
over subsets of $E$ such that each element $e \in E$ is included in a sample, independently, with
probability $p_e$. For a collection $\calS = \{S_1, \dots, S_k\}$ such that $S_i \in \calI, \forall i \in \{1, \dots, k\}$, we define its value as

\begin{align*}
    \val(\calS) = \E_{A\sim\calD}\left[\max_{S_i \in \calS} \left|S_i \cap A \right|\right].
\end{align*}

\subsection{Preliminaries}

In this subsection we introduce notation and give some preliminary definitions. First, for any positive integer $\ell$ we use the shorthand
$[\ell]$ to denote the set $\{1, \dots, \ell\}$. In addition, we give the definition of matroids below.

 \begin{definition}\label{def: matroid} A pair $M=(E, \calI)$ is called a matroid if $E$ is a finite ground set and $\calI$ is a non-empty collection of subsets of $E$ such that
 \begin{enumerate}
     \item If $I \in \calI$ and $J \subseteq I$, then $J \in \calI$,
     \item If $I, J \in \calI$ and $|I| < |J|$, then $I + e \in \calI$ for some $e \in J \setminus I$.
 \end{enumerate}
 \end{definition}

\noindent For a matroid $M = (E, \calI)$ the elements of $\calI$ are called the independent sets of $M$. 
The rank function of a matroid $M$, $r_M : 2^E \rightarrow \nats$, maps each set $U \subseteq E$ 
to the size of the largest independent set contained in $U$. We use the term ``rank of a matroid $M$''
to denote the rank $r$ of the ground set, i.e. $r = r_M(E)$. 
A set $B \subseteq E$ is called a base if and only
if it's a maximum size independent set. In addition, the span of a set $S \subseteq E$ is defined as
$\spann_M(S) = \{e \in E: \rank(S\cup\{e\}) = \rank(S)\}$.

\vspace{1em}
The matroid polytope $\calP(M)$, of a matroid $M$, is a subset of $\reals^{|E|}$ that is defined by the following sets of inequalities, where $r_M$ is the rank function of $M$.
\[
\calP(M) = \left\{
\begin{aligned}
     & \sum_{e \in U} x_e \leq r_M(U), \quad &\forall U \subseteq E \\
     & x_e \geq 0, \quad &\forall e \in E
\end{aligned}
\right \}
\]

 It is well known that one can check whether some $x \in \mathbb{R}^{|E|}$ lies in the matroid polytope, that is $x \in \calP(M)$, in polynomial time in the size of $E$. Finally, we state some known results about matroids, that
 we use in our analysis, in Appendix~\ref{subsec:appendix-matroids}.

\section{Overview of our Techniques}\label{sec: Our Techniques}

The core difficulty of the \matroidproblemname problem stems from its objective function. Indeed, calculating or 
bounding, the expectation of the maximum of $k$ random variables is a difficult task, especially when the 
random variables are dependent. Apart from that, optimizing under this objective requires  to look
at the problem from a perspective that differs with what we are used to in the 
analysis of randomized algorithms. Usually, we argue that our solutions have a good enough 
expectation and that they reach this expectation with a reasonable probability. 
However, in order to construct a good portfolio we are interested in the ``anti-concentration'' properties of the  
solutions that we pick, meaning that we want each solution to be able to greatly surpass its expected value 
with a reasonable probability. 
Intuitively, if the  values of our solutions were i.i.d. random variables, we would want them to be 
somewhat ``heavy-tailed'' so that after $k$ independent samples there would be a good probability that we 
observe a high outlier. On the other hand, if the solutions that we pick were highly concentrated
around their expected value, then we would observe almost no benefit by taking the maximum of several random 
variables.

In this section, we first highlight our ideas for the simpler case of \matroidproblemname for uniform matroids and then explain how they can properly be generalized to work for all matroids.
As a reminder, in the \matroidproblemname problem on uniform matroids, we are given
a ground set of $n$ elements and we want to construct a portfolio $\calP$ consisting of $k$ subsets of the elements, each of size $r$. 
The value of the $i$-th element is an independent Bernoulli random variable with probability $p_i$ and we are interested in maximizing the value of the portfolio that we construct.

Even in this simple case, understanding the core trade-off of using a higher 
probability element multiple times across solutions versus replacing it with independent elements of
lower probability is a challenging task. In order to build some intuition, in the next subsection,
we discuss some natural approaches for this simple case and show why they fail to produce a portfolio 
that is a constant-factor approximation of the optimal portfolio. 

\subsection{Natural approaches that fail} 

The first approach for this problem is simply to pick the $k$ 
independent sets with the highest expected value. Of course, it makes sense to enforce that these
independent sets are \textit{disjoint}, since we want them to ``complement'' each other. If these sets had large pairwise intersections, we would observe no benefit from selecting $k$ of them instead of one. In other words, a natural strategy is to pick the highest expectation subset,
remove it from the ground set, continue by picking the highest expectation subset of the remaining
elements and so on. However, this approach will fail to construct a good portfolio. 

Intuitively, 
in some cases it is better to restrict ourselves to considerably less disjoint subsets and 
complete our portfolio by combining those subsets in a clever way, instead of using new subsets that might 
have lower expectation. For example, consider the instance where
we need to pick $k$ subsets of size $r = k$ and we have $n = r^2$ elements available. Let the first $ k \log k$
elements have activation probability $\nicefrac{1}{k}$, and the rest to have activation probability 
$\nicefrac{1}{k^2}$. The approach described above will form $k$ disjoint subsets, out of which $\log k$ will 
behave as independent binomials with expectation $1$ and the rest will be independent binomials
with expectation $\nicefrac{1}{k}$. The value of this portfolio will be dominated by 
the expectation of the maximum  of the first $\log k$ binomials, which is $O(\log \log k)$.
Surprisingly, instead of using the lower probability
elements, one can pick $k$ solutions by uniformly combining parts of the first $\log k$ subsets
and construct a portfolio that has value $\Theta(\log k / \log \log k)$, which is asymptotically
optimal (we describe this construction in Appendix~\ref{sec:appendix-bruteforce-construction}). 
This example showcases the power of having a clever ``mixing'' strategy
as this can boost the portfolio to achieve an exponentially better value than a portfolio
that only uses disjoint solutions. In fact, in this example, a strategy that picks only disjoint
solutions would need $k$ disjoint subsets of expectation $1$ in order to achieve the same value as
the one achieved by uniformly mixing $\log k$ subsets of expectation $1$.



\subsection{Main ideas}

\paragraph{Filtering out elements.} From the previous example, it becomes evident that a candidate 
algorithm should have a filtering procedure that discards some elements of the ground set,
and a mixing strategy, which combines the remaining elements to form the desired subsets.
It is natural to wonder if one can commit to the simplest possible mixing strategy
(i.e. sampling elements uniformly) and try to find a filtering rule that would make this strategy
work. Note that the nearly-optimal solution constructed for the previous example actually fits into this 
framework.

In the simpler case of uniform matroids, designing the filtering procedure can be reduced to first 
sorting the elements in decreasing order of their probabilities and then selecting a prefix of this 
order. A key observation in our work
is that there always exists a prefix of the elements that admits a good portfolio under the simplest
mixing strategy, i.e. forming subsets by uniform sampling. 
This observation directly gives us a polynomial-time constant-approximation algorithm, as one can
try all possible $n$ prefixes, generate the corresponding portfolios through uniform sampling and keep the best 
one of them after estimating their values. We formally present this algorithm in Section~\ref{sec: Uniform Matroids} before building up to our algorithm for general
matroids in Section~\ref{sec:matroids-algorithm}.

\paragraph{Analyzing the value of the produced portfolio.} The simplicity of our mixing strategy allows
us to bypass dealing with the dependencies of the solutions we form when we want to lower bound
the expectation of their maximum. On a high level, to analyze the value of our portfolio, we first fix the 
randomness of the instance, by conditioning on an outcome for the active elements,
and then analyze the expected value of a sampled solution over the internal randomness of the algorithm.
Once we have fixed the outcome of the active elements, the \textit{values} of the sampled solutions
are simply binomial random variables that only depend on the internal randomness of the sampling procedure.
Critically, this makes the values of the sampled solutions
\textit{independent} random variables which makes them much easier to work with.
On the other hand, if we had first fixed the portfolio produced by the algorithm and then tried
to analyze its value over the randomness of the instance, we would have to analyze the dependencies
of the picked solutions and how these influence the value of the portfolio, which is a much harder task.
This analysis technique, of course, comes with its own challenges as selecting the appropriate event
to condition on is not a trivial task.

\paragraph{Generalizing to all matroids.} The filtering procedure described above fails to work beyond uniform matroids, simply because it ignores the underlying matroid structure of the elements. 
For example, consider the case of the graphic matroid of a graph $G$ that consists of a clique on $\sqrt{n}$ 
vertices and a simple path of $n-\sqrt{n}$ vertices. Let the activation probabilities of the edges of the clique to be
$\nicefrac{1}{\sqrt{n}}$ and the activation probabilities of the edges of the path to be $\nicefrac{1}{\sqrt{n}}-
\epsilon$ for some small $\epsilon > 0$. 
Ordering the elements by their activation probabilities and selecting a prefix of
this ordering goes towards the wrong direction as one should prioritize taking edges from the path over taking 
edges from the clique.
Indeed, taking the path as our only solution would give us an expected value of $\Theta(\sqrt{n})$, whereas
any portfolio of $k = O(n)$ spanning trees of the clique has value at most $O(\log n / \log \log n)$.

In order to overcome this issue, we change our filtering procedure to select a certain number of disjoint, high-expectation, independent sets of the matroid. In other words, we create an ordering of disjoint independent
sets of decreasing expectation and we take a prefix of this ordering. Constructing this ordering can simply be achieved by finding the biggest expectation base of the matroid through the Greedy algorithm,
then restricting the matroid to the remaining elements and repeating. As in the case of uniform matroids, 
we prove that there always exists a prefix of this ordering that admits a nearly-optimal portfolio after 
sampling elements uniformly and then passing them through a contention resolution scheme. 

Finally, the biggest technical challenge of this problem is analyzing the value of the produced portfolios
in the case of general matroids. To be more precise, the value of a portfolio heavily relies on
``abnormal'' events happening in one of its $k$ solutions. For example, if we fix a set of active elements, and analyze the expected value of the sampled solutions over the
internal randomness of the algorithm, the expected value of each solution will be the
number of active elements we expect to sample in one trial. Crucially, one of the $k$
trials will sample much more active elements than its expectation, and this trial will be
responsible for the value of the portfolio.
At this point, this solution will be passed through 
a contention resolution scheme to be trimmed down to an independent set. However, contention resolution
schemes do not have per-instance guarantees but only work in expectation. Therefore, there is no guarantee
that, when this abnormal event happens, the active elements are not going to be discarded by the contention
resolution scheme. We overcome this issue by conditioning on appropriate events that do not entirely fix
the randomness of either the instance or the algorithm, but still allow us to use 
the guarantees of a contention resolution scheme and argue that one of the sampled solutions
reaches a near-optimal value. This is the main result of our work which is formally stated below.





\vspace{2em}

\begin{restatable}{theorem}{maintheorem}\label{thm: main theorem}
If $\calD$ is a product distribution then Algorithm~\ref{alg:matroids} is a $\Theta(1)$-approximation algorithm for the \textit{$k$-portfolio solution} problem and has a polynomial time complexity.
\end{restatable}

\section{Warm-up: An $O(1)$ approximation for Uniform Matroids}\label{sec: Uniform Matroids}

In this section we formally introduce some of the ideas described above by presenting our algorithm for the case of uniform matroids.
We remind the reader that the input is described by a triplet $\calJ = (M, k,\calD)$ where $M=(E, \calI)$ is a (uniform in this section) matroid with rank $r$, $k$ is the size of our portfolio and $\calD$ is a product distribution where each element $e \in E$ is active with probability $p_e$. For simplicity, we order elements in $E$ in decreasing order of their probabilities. We denote by $e_i$ the $i$-th element in this order and by $p_i$ its activation probability. In addition, throughout this section we assume that $\OPT \geq 200$. In~\cref{sec:matroids-algorithm}, we prove that lower bounding $\OPT$ by a large constant is without loss of generality.

As mentioned in Section~\ref{sec: Our Techniques}, our algorithm filters out some elements
of the ground set and then produces the $k$ solutions of its portfolio
by sampling elements uniformly at random. We prove that for any instance of the problem,
there exists a prefix of the elements (when ordered by decreasing activation probability)
that admits a near-optimal portfolio under uniform sampling. Therefore, the algorithm
simply needs to try all $n$ prefixes, estimate the values of the constructed portfolios
and output the best one of them. We give the pseudocode of this strategy in 
Algorithm~\ref{alg:uniform}.

\begin{algorithm}[h]
\caption{An algorithm for Uniform Matroids}\label{alg:uniform}
\begin{algorithmic}

\Function{Create-Portfolio}{$\calJ=(M,k,\calD)$} \Comment{$M$: uniform matroid of rank $r$}
\State $\text{Portfolios} \gets []$
\For{$i = 1, \dots ,n$}
    \State $\text{Prefix} \gets \{1, \dots, i\}$ \Comment{Elements are ordered with decreasing $p_i$}
    \State $\text{Portfolios}[i] \gets \Call{Portfolio-From-Prefix}{\text{Prefix},k,r}$
\EndFor
\State Estimate the values of the $n$ portfolios
\State Return the portfolio with the biggest estimated value
\EndFunction

\State 

\Function{Portfolio-From-Prefix}{Prefix, $k$, $r$}
\State $\calP \gets \{\}$
\For{$i = 1,\dots,k$}
    \State Let $V_i = \{s_1, \dots s_r\}$ be $r$ uniformly random samples from Prefix
    \State $S \gets \text{non-duplicate elements of $V_i$}$
    \State $\calP \gets \calP \cup S$
\EndFor
\State
\Return $\calP$
\EndFunction

\end{algorithmic}
\end{algorithm}

\vspace{1em}
In order to prove that Algorithm~\ref{alg:uniform} is an $O(1)$-approximation, it suffices
to show that there exists a prefix on which sampling elements uniformly creates a portfolio
with value at least $\Theta(1) \cdot \OPT$. 
Indeed, the value of each solution can be estimated using standard techniques within a small factor with polynomially many samples with high probability. Since there is only $n$ solutions the estimate is close for all of them with good probability (by union bound) and we output the best solution.
The prefix on which we will focus, is the largest prefix whose expected value is at most $\OPT / 2$. More formally, let $M$ be the largest index
such that
$$
\sum_{i = 1}^M p_i < \OPT /2.
$$

\noindent By the definition of $M$, we know that $\sum_{i = 1}^{M+1} p_i \geq \OPT / 2$, therefore we also
get that
$$
\sum_{i = 1}^M p_i \geq \OPT /2 - p_{M+1} \geq \OPT /2 - 1 \geq \OPT/3,
$$
\noindent where in the last inequality we used that $\OPT \geq 6$.

\vspace{1em}

Our first observation is that there exists a portfolio which only uses elements 
outside of the 
selected prefix, i.e. elements with lower activation probability
than what our algorithm has picked, and that achieves value at least $\OPT / 2$.
This portfolio is the restriction of the optimal one to the elements outside of the prefix.
Let $H = \{e_1, \dots, e_M\}$ and $L = \{e_{M+1}, \dots, e_n\}$. The aforementioned claim
is stated formally in the following lemma:

\begin{lemma}\label{lem:uniform-opt-outside-prefix}
    $\E_{A\sim \calD}\left[\max_{O_i \in \calO} |O_i \cap L \cap A|\right] \geq \frac{\OPT}{2}$.
\end{lemma}

\begin{proof}[Proof of Lemma~\ref{lem:uniform-opt-outside-prefix}]
    \begin{align}
    \OPT &= \E_{A\sim \calD}\left[\max_{O_i \in \calO} |O_i \cap A|\right]\\
         &= \E_{A\sim \calD}\left[\max_{O_i \in \calO} |O_i \cap H \cap A| + |O_i \cap L \cap A|\right]\\
         &\leq \E_{A\sim \calD}\left[ |A \cap H| + \max_{O_i \in \calO} |O_i \cap L \cap A|\right]\\
         &\leq \frac{\OPT}{2} + \E_{A\sim \calD}\left[\max_{O_i \in \calO} |O_i \cap L \cap A|\right]\\
    \Rightarrow & \E_{A\sim \calD}\left[\max_{O_i \in \calO} |O_i \cap L \cap A|\right] \geq \frac{\OPT}{2},
\end{align}

\noindent where from 
(2) to (3) we used that $O_i \cap H \cap A \subseteq H \cap A$ for all $i$ and from 
(3) to (4) that $\E_{A\sim \calD}\left[|A \cap H|\right] =
\sum_{i=1}^M p_i \leq \frac{\OPT}{2}$.
\end{proof}

The next observation we make is that if someone had access to $k\cdot r$ independent copies of the element
$e_{M+1}$, then by using these elements they could construct a portfolio that has value at least 
$\OPT / 2$.

\begin{lemma}\label{lem:uniform-independent-M+1}
    Let $B_1, \dots, B_k$ be $k$ i.i.d random variables following $\Bin(r,p_{M+1})$. Then,
    $$\E\left[\max_{i \in [k]} B_i\right] \geq \frac{\OPT}{2}.$$
\end{lemma}

\begin{proof}[Proof of Lemma~\ref{lem:uniform-independent-M+1}]
    From Lemma~\ref{lem:uniform-opt-outside-prefix}, we know that there exist $k$ dependent Poisson Binomials,
    namely the solutions picked by the optimal portfolio restricted to $L$, 
    each of which has at most $r$
    trials, trial probabilities at most $p_{M+1}$, and whose expected maximum is at least $\OPT/2$. 
    The proof follows from the fact that one construct the desired
    Binomials, $B_1, \dots, B_k$, by starting from the Poisson Binomials and doing the following transformations:
    \begin{enumerate}
        \item Increase the number of trials of all Poisson Binomials to $r$ by augmenting new independent Bernoulli random variables each having probability $p_{M+1}$.
        \item Make the Poisson Binomials independent by introducing new independent copies for the elements
        that are shared across many solutions.
        \item Transform the Poisson Binomials into Binomials by increasing all probabilities to $p_{M+1}$.
    \end{enumerate}

    \noindent All of the above transformations do not decrease the expectation of the maximum of the random
    variables, since (1) the augmented variables stochastically dominate the previous ones, (2)
    making the random variables independent can only increase the expectation of their maximum, because solutions with shared elements are positively correlated and ``fail'' together
    (Lemma~\ref{lem:binomials-independent-better}) and (3) increasing the probabilities of the trials
    can also only increase the expectation of their maximum (Lemma~\ref{lem:binomials-bigger-p-better}).
\end{proof}

\vspace{1em}
The previous lemma gives us a  ``target'' for analyzing the expected value of our portfolio. After conditioning on an appropriate constant probability event for the activation of elements, we argue that the \textit{values} of the sampled solutions are independent Binomial random variables with trial probability close to $p_{M+1}$.
We define such event as $H$ having at least $\OPT/12$ active elements and prove that the latter happens with probability at least $1/2$.
To that end, let $W = \left \{\tilde{E} \subseteq E: |\tilde{E}\cap H| \geq \OPT / 12 \right \}$. 

\begin{lemma}\label{lem:uniform-condition-events-prob}
    $\Pr_{A \sim \calD}\left[ A \in W \right] \geq 1/2$.
\end{lemma}

\begin{proof}[Proof of Lemma~\ref{lem:uniform-condition-events-prob}]
    The left hand side of the desired inequality can be written as:

\setcounter{equation}{0}
\begin{align}
    \Pr_{A \sim \calD}\left[ A \in W \right] & = \Pr_{A \sim \calD}\left[ |A\cap H| \geq \frac{\OPT}{12} \right]\\ 
    & \geq  \Pr_{A \sim \calD}\left[ |A\cap H| \geq \frac{\E \left[ |A \cap H| \right]}{4} \right]\\
    & \geq  \Pr_{A \sim \calD}\left[ \left| |A\cap H| - \E \left[ |A \cap H| \right] \right|  \leq \frac{3}{4}\E \left[ |A \cap H| \right] \right]\\
    & \geq 1-2e^{-\left(\frac{3}{4}\right)^2 \cdot \frac{1}{3} \cdot \E \left[ |A \cap H| \right] }\\
    & \geq \frac{1}{2}
\end{align}

\noindent where for (1) and (4) we used that $\E_{A\sim \calD}\left[|A \cap H|\right] =
\sum_{i=1}^M p_i \geq \frac{\OPT}{3} \geq 10$ and for (3) we used a Chernoff Bound (Corollary~\ref{cor:chernoff})
for the Binomial random variable $|A \cap H|$.
\end{proof}

\vspace{1em}
\noindent We continue by proving~\cref{lem:uniform-multisets-opt} for the expected value of the sampled multisets $V_i$. For simplicity, we slightly abuse notation and for any set $\tilde{E} \subseteq E$ we use $|V_i \cap \tilde{E}|$ to denote the sum $\sum_{x \in V_i} \indic\{x \in \tilde{E}\}$.

\begin{lemma}\label{lem:uniform-multisets-opt}
   For any $\tilde{E} \in W$, it holds that $$\E_{\ALG} \left[ \max_{i \in [k]} |V_i \cap \tilde{E}| \right] \geq \OPT/24.$$
\end{lemma}

\begin{proof}[Proof of Lemma~\ref{lem:uniform-multisets-opt}]
    
When we have fixed an outcome $\tilde{E} \in W$ for the active elements, the values of the sampled 
multisets, $V_i$, depend only on the internal randomness of our sampling procedure and are, thus, 
independent random variables. In addition, the algorithm samples elements uniformly at random from $H$. Therefore the probability that a sampled element is active is
$$
\frac{|\tilde{E} \cap H|}{|H|} \geq
\frac{\OPT}{12 \cdot |H|} \geq \frac{\sum_{i \in H} p_i}{12\cdot |H|} \geq \frac{p_{M+1}}
{12},
$$

\noindent where for the first inequality we used that the definition of $ W$, for the second inequality we used that $\sum_{i \in H} p_i \leq \OPT/2$ from the definition of
the prefix $H$, and for the third inequality we used that $\forall i \in H: p_i \geq p_{M+1}$. 

\vspace{1em}
 Therefore, the random variables $|V_i \cap \tilde{E}|$
are independent Binomials with trial probability at least $p_{M+1}/12$ for all $i$.
Intuitively, this means that the expectation of their maximum should be close to the expectation
of the maximum of $k$ independent binomials with trial probability $p_{M+1}$, which by 
Lemma~\ref{lem:uniform-independent-M+1} is at least $\Theta(1) \cdot \OPT$.

\vspace{1em}
\noindent More formally, let $B_1, \dots, B_k$ be iid random variables following $\Bin(r,p_{M+1})$
and $B_1', \dots, B_k'$ be iid random variables following $\Bin(r,p_{M+1}/12)$.
Then for any $\tilde{E} \in W$, it holds that

$$
    \E_{\ALG} \left[ \max_{i \in [k]} |V_i \cap \tilde{E}| \right] \geq 
    \E \left[\max_{i \in [k]} B_i'\right]
    \geq \frac{1}{12} \E \left[\max_{i \in [k]} B_i\right]
    \geq \frac{\OPT}{24},
$$
\vspace{1em}

\noindent where the first inequality holds because the values of the sampled multisets dominate the random variables $B_i'$ (Lemma~\ref{lem:binomials-bigger-p-better}). For the second inequality,
intuitively, one can view the process of sampling the random variables $B_i'$ as first sampling the Binomials
$B_1, \dots, B_k$ and then discarding every Bernoulli random variable that succeeded, independently,
with probability $1/12$. In this way, for every outcome of the Binomials $B_1, \dots B_k$,
the Binomials $B_1', \dots, B_k'$ will retrieve, on expectation, a $1/12$ factor of their
corresponding random variables $B_i$ (Lemma~\ref{lem:scaled_prob_bin}).
Finally, for the third inequality we used Lemma~\ref{lem:uniform-independent-M+1}.

\end{proof}

We continue by proving that for any ``good'' outcome $\tilde{E} \in W$, the expected value of the maximum
of the sets $S_1, \dots, S_k$, that consist of the unique elements of the sampled multisets, $V_i$,
is still $\Theta(1) \cdot \OPT$.

\begin{lemma}\label{lem:uniform-exp-value-on-good-events}
   For any $\tilde{E} \in W$, it holds that
   $$
   \E_{\ALG} \left[ \max_{i \in [k]} |S_i \cap \tilde{E}|\right] \geq \frac{e}{240(e-1)} \cdot \OPT.
   $$
\end{lemma}

\begin{proof}[Proof of Lemma~\ref{lem:uniform-exp-value-on-good-events}]
    
In order to analyze the expected value of the portfolio for the events $\tilde{E} \in W$, 
we will condition on the value of the maximum of the multisets $V_1, \dots, V_k$. By the law of total
expectation we get that

\begin{align*}
    \E_{\ALG} \left[ \max_{i \in [k]} |S_i \cap \tilde{E}| \right] =
\sum_{x} \E_{\ALG} \left[ \max_{i \in [k]} |S_i \cap \tilde{E}| \Bigg | \max_{i \in [k]} |V_i \cap \tilde{E}| = x \right] \cdot \Pr_{\ALG}\left[ \max_{i \in [k]} |V_i \cap \tilde{E}| = x \right].\tag{$\ddagger$} \label{eq:uniform-law-total-exp}
\end{align*}

\vspace{1em}
\noindent Lower bounding the following expression,
$$
\E_{\ALG} \left[ \max_{i \in [k]} |S_i \cap \tilde{E}| \Bigg | \max_{i \in [k]} |V_i \cap \tilde{E}| = x \right], 
$$

\noindent can be seen as a balls and bins question.
Specifically, we know that one of the $k$ solutions formed by the algorithm sampled $x$ items from 
$|\tilde{E} \cap H|$.
Since the algorithm is sampling elements uniformly at random, those $x$ items are also distributed
uniformly at random inside $|\tilde{E} \cap H|$. We are interested in calculating the expected number
of distinct items that were sampled. This question is equivalent to counting the non-empty bins
after throwing $x$ balls uniformly into $|\tilde{E} \cap H|$ bins. 
Using a standard result about the expected number of non-empty bins (Lemma~\ref{lem: lower bound on expected number of full bins }), we get that

$$
\E_{\ALG} \left[ \max_{i \in [k]} |S_i \cap \tilde{E}| \Bigg | \max_{i \in [k]} |V_i \cap \tilde{E}| = x \right] \geq \min \left\{ \frac{x}{2}, \frac{3|\tilde{E}\cap H|}{10} \right\}.
$$

\noindent Eq.~\ref{eq:uniform-law-total-exp} can now be re-written as

\setcounter{equation}{0}
\begin{align}
    \E_{\ALG} \left[ \max_{i \in [k]} |S_i \cap \tilde{E}| \right] &\geq
    \sum_{x} \min \left\{ \frac{x}{2}, \frac{3|\tilde{E}\cap H|}{10} \right\}
    \cdot \Pr_{\ALG}\left[ \max_{i \in [k]} |V_i \cap \tilde{E}| = x \right]\\
    & \geq \sum_x \min \left\{ \frac{x}{2}, \frac{\OPT}{40} \right\} \cdot \Pr_{\ALG}\left[ \max_{i \in [k]} |V_i \cap \tilde{E}| = x \right]\\
    & \geq  \min \left\{ \frac{\OPT}{300}, \frac{\OPT}{40} \right\} \cdot \Pr_{\ALG}\left[ \max_{i \in [k]} |V_i \cap \tilde{E}| \geq \frac{\OPT}{150} \right]\\
    & \geq \frac{\OPT}{300} \cdot \Pr_{\ALG}\left[ \max_{i \in [k]} |V_i \cap \tilde{E}| \geq \frac{\E_{\ALG} \left[ \max_{i \in [k]} |V_i \cap \tilde{E}| \right]}{5} \right]\\
    &\geq \frac{e}{300(e-1)} \cdot \OPT,
\end{align}

\noindent where for (2) we used the fact that $|\tilde{E} \cap H| \geq \OPT/12$ for $\tilde{E} \in W$,
to get (3) we restricted the sum to the terms $x \geq \OPT / 150$, to get (4) we used that $\E_{\ALG} \left[ \max_{i \in [k]} |V_i \cap \tilde{E}| \right] \geq \OPT/24$ from Lemma~\ref{lem:uniform-multisets-opt} and
to get (5) we used a concentration inequality for the maximum of independent Binomial random variables
(Lemma~\ref{lem:max_binom_concentration}) and the assumption that $\OPT \geq 800$ to get that 
$\E_{\ALG} \left[ \max_{i \in [k]} |V_i \cap \tilde{E}| \right]$ is at least some constant.

\end{proof}

\noindent Finally, we are ready to prove the main theorem of the section.

\begin{theorem}\label{thm:uniform-main}
    Algorithm~\ref{alg:uniform} is a $\Theta(1)$-approximation algorithm for the \textit{$k$-solution portfolio} problem when the given matroid is uniform and $\calD$ is a product distribution.
\end{theorem}

\begin{proof}[Proof of Theorem~\ref{thm:uniform-main}]
    The value of the constructed portfolio $\calP$ can be written as:

\setcounter{equation}{0}
\begin{align}
    \val (\calP) &= \E_{\ALG, A \sim \calD} \left[ \max_{i \in [k]} |S_i \cap A| \right] \\
    & \geq \sum_{\tilde{E} \in W} \E_{\ALG} \left[ \max_{i \in [k]} |S_i \cap \tilde{E}|\right] \cdot
    \Pr_{A \sim \calD}\left[ A = \tilde{E}\right]\\
    & \geq \frac{e}{300(e-1)} \cdot \OPT \cdot \sum_{\tilde{E} \in W} \Pr_{A \sim \calD}\left[ A = \tilde{E}\right]\\
    & \geq \frac{e}{300(e-1)} \cdot \OPT \cdot \Pr_{A \sim \calD}\left[A \in W\right]\\
    & \geq \frac{e}{600(e-1)} \cdot \OPT,
\end{align}

\noindent where to get (3) we applied Lemma~\ref{lem:uniform-exp-value-on-good-events} and (5) follows from Lemma~\ref{lem:uniform-condition-events-prob}.
\end{proof}
\section{An $O(1)$ approximation for all Matroids}\label{sec:matroids-algorithm}

To make the algorithm's description simpler, we start by introducing a series of simplifying 
assumptions which do not change the complexity of our problem, this means that we can get a constant-factor
approximation even without these assumptions). Let $\calJ = (M, k, \calD)$ be a 
triplet describing the input and $\ell \geq k$ a positive integer. Then, without loss of generality, we assume that:

\newlist{Assumptions}{enumerate}{2}
\setlist[Assumptions]{label=(Assumption \arabic*), font=\textbf, itemindent=2.2cm}

\begin{Assumptions}\label{matroid-assumptions}
        \item $M$ has at least $\ell$ disjoint bases.\label{assum:M-has-ell-disjoint-bases}\label{assum: M has l disjoint bases}
        \item The expected value of any independent set is at most $\OPT /4$. \label{assum:maximum-probability-base-is-at-most-OPT/4}
        \item The value of the optimal portfolio is $\OPT \geq 4100$.
\end{Assumptions}

\paragraph{Assumption 1.} If the first assumption is not true for the matroid of the input, then for each element $e$ we create $\ell-1$ duplicates, each of which has $0$ probability of being active. Let $M'$ be the new matroid, $\calD'$ be the new product distribution which includes the duplicate elements, $\calJ = (M, k, \calD)$ the initial input and $\calJ' = (M', k, \calD')$ our modified input. It is trivial to see that  $$\val \left(  \calO (\calJ) \right) = \val \left(  \calO (\calJ') \right),$$
\noindent where $\calO (\calJ)$ and $\calO (\calJ')$ are the optimal portfolios of the instances
$\calJ$ and $\calJ'$.

\vspace{1em}
In addition, for any solution $\calS'$ for input $\calJ'$ we can compute a solution $\calS$ for input $\calJ$ by replacing each duplicate element with the original one such that $\val ( \calS ) \geq \val ( \calS' )$. Thus, any constant-factor approximation algorithm when the input matroid has $\ell$ disjoint bases can be translated to a constant-factor approximation algorithm without this assumption.

\paragraph{Assumption 2.} If the second assumption is not true then $\exists I \in \calI$ such that 
$\E_{A \sim \calD} [|I \cap A|] \geq \OPT/4$. In that case, any portfolio $\calS$ which contains $I$ is a 
$1/4$-approximation to the optimum value. Furthermore, in this case, the highest-expectation
base, let it be $B_1$, will also have an expected value that is greater than $\OPT/4$. Therefore, 
there is no need to test if this assumption holds, since one can run their algorithm to construct a 
portfolio $\calS$ and simply exchange one of the portfolio's sets with $B_1$ if the portfolio 
$\calS$ happens to have a lower value than the expectation of $B_1$.

\paragraph{Assumption 3.} If $\OPT < 4100$ then we argue that the portfolio which contains the $k$ disjoint bases of maximum weight is a constant factor approximation to the optimal portfolio. Let those bases be $B_1, \dots, B_k$, $\mu = \sum_{e \in \bigcup_{i = 1}^k B_i} p_e $ and note that $\val (\calO ) \leq \sum_{e \in \bigcup_{i = 1}^k O_i} p_e \leq \mu$.
Thus, $ \val (\calO ) \leq \min \left\{  \mu, 4100 \right\} $. In~\cref{lem: assumption 3} we prove that if $\mu > 1/2$ then the probability of $\bigcup_{i = 1}^k B_i$ containing an active element is $\Omega(1)$ and if $\mu \leq 1/2$ then that probability is $\Omega(\mu)$. Thus, in each case it holds that: $\val (\calB)\geq \Pr \left[ A \cap \bigcup_{i = 1}^k B_i \neq \emptyset\right] \geq \Theta(1) \cdot \val (\calO )$.

\vspace{1em}
\subsection{Description of the algorithm} 

Our algorithm (Algorithm~\ref{alg:matroids}) is
similar to the one for uniform matroids (Algorithm~\ref{alg:uniform}). As we have discussed
in Section~\ref{sec: Our Techniques}, the prefix used by Algorithm~\ref{alg:uniform} cannot
work for general matroids. Instead, we define a new notion of a prefix; we create an order
of disjoint independent sets that are sorted by their expected value. To construct this order, 
the algorithm continuously finds the independent set with the highest expectation, removes it from
the matroid and puts it next in the ordering. Due to~\ref{assum: M has l disjoint bases},
we can safely assume that the independent sets picked will be bases of the given matroid, potentially including some artificial elements with zero activation probability.

Having constructed the aforementioned ordering, Algorithm~\ref{alg:matroids} proceeds
by creating two portfolios for every prefix of this ordering. It then, using standard techniques, estimates the values of each constructed portfolio within a small factor with high probability only requiring  polynomially many samples from distribution $\calD$. Since there are only $2n$ portfolios the estimate is close for all of them with good probability (by union bound) and we output the portfolio with the highest estimated value.

To understand how the portfolios of a given prefix are created, let's focus
on the $i$-th prefix, i.e. the first $i$ bases of the constructed order.
Due to standard results in matroid theory and the properties of our order, we show that
every element of the $i+1$-st base can be mapped into $i$ distinct elements in the prefix
that have greater or equal activation probability. We call this group of $i$ elements a ``column''
and we also use the term ``Column-Decomposition'' to refer to the partitioning of the prefix
into columns. We present how to construct the ``Column-Decomposition'' in 
Section~\ref{subsec:column-decomposition}. After creating this partitioning, Algorithm~\ref{alg:matroids}
proceeds by constructing the two following portfolios
\begin{itemize}
    \item \textbf{Uniform Portfolio}: For each base of the uniform portfolio, we sample $r$ uniformly random elements from the prefix and then pass them through a standard contention resolution scheme.
    This process is described in Algorithm~\ref{alg:matroids-uniform}.

    \item \textbf{Column Portfolio}: Each base of the column portfolio is constructed as follows. For each element $e$ of the $i+1$\nobreakdash-st base in the order (i.e. the highest expectation base outside of the prefix), we select one element from its column uniformly at random. Then, we pass the the sampled
    set through a standard contention resolution scheme. This process is described in Algorithm~\ref{alg:matroids-column}.
\end{itemize}

Before heading to the algorithms' pseudocodes, we want to remark that constructing only the uniform portfolios suffices in order
for Algorithm~\ref{alg:matroids} to be a $\Theta(1)$-approximation. In fact, the portfolios constructed from the two processes
described above will be very similar, since the uniform strategy, with good probability, will also sample one uniform
element from most columns. However, introducing the column portfolios makes the analysis much simpler.

\begin{algorithm}[H]
\caption{An algorithm for General Matroids}\label{alg:matroids}
\begin{algorithmic}

\Function{Create-Portfolio}{$\calJ$}
\State $\text{Portfolios} \gets []$
\State $\text{Ordering} \gets []$
\State $E_r \gets E$
\State {\color{blue}{*Create the ordering of bases*}}
\For{$i = 1, \dots ,n$}
    \State $B_i \gets \text{highest expectation base in } E_r$
    \State $E_r \gets E_r \setminus B_i$
    \State $\text{Ordering}[i] \gets B_i$
\EndFor
\State {\color{blue}{*Try all possible prefixes*}}
\For{$i = 1, \dots ,n$}
    \State $\text{Prefix} \gets \text{Ordering}[1\dots i]$
    \State $\text{Portfolios.append(}\Call{Uniform-Portfolio}{\text{Prefix}, \calJ})$ \Comment{Algorithm~\ref{alg:matroids-uniform}}
    \State $\text{Portfolios.append(}\Call{Column-Portfolio}{\text{Prefix}, \calJ})$
    \Comment{Algorithm~\ref{alg:matroids-column}}
    
\EndFor
\State Estimate the values of the $2n$ portfolios
\State Return the portfolio with the biggest estimated value
\EndFunction
\end{algorithmic}
\end{algorithm}

\begin{algorithm}[H]
\caption{Portfolio from uniform sampling}\label{alg:matroids-uniform}
\begin{algorithmic}
\Function{Uniform-Portfolio}{Prefix, $\calJ$}
\State $\calP \gets \{\}$
\For{$j = 1,\dots,k$}
    \State Let $\widetilde{V}$ contain $r$ uniformly random samples, with replacement, from  Prefix
    \State $V \gets \Call{Remove-duplicates}{\widetilde{V}}$
    \State $S \gets \Call{Contention-Resolution-Scheme}{V}$ \Comment{Defined in Section~\ref{subsec:crs-main-body}}
    \State $\calP \gets \calP \cup S$
\EndFor
\State
\Return $\calP$
\EndFunction
\end{algorithmic}
\end{algorithm}

\begin{algorithm}[H]
\caption{Portfolio from column-wise sampling}\label{alg:matroids-column}
\begin{algorithmic}
\Function{Column-Portfolio}{Prefix, $k$, $i$}
\State $\calP \gets \{\}$
\State $C_1, \dots, C_r \gets \Call{Column-Decomposition}{\text{Prefix}, i}$ \Comment{Defined in Section~\ref{subsec:column-decomposition}}
\For{$j = 1,\dots,k$}
    \State Let $V$ contain one uniformly random element from each column $C_i$
    \State $S \gets \Call{Contention-Resolution-Scheme}{V}$ \Comment{Defined in Section~\ref{subsec:crs-main-body}}
    \State $\calP \gets \calP \cup S$
\EndFor
\State
\Return $\calP$
\EndFunction
\end{algorithmic}
\end{algorithm}

\subsection{Contention Resolution Scheme} \label{subsec:crs-main-body}
In this section we describe the properties that the contention resolution scheme (CRS) used by Algorithms~\ref{alg:matroids-uniform} and~\ref{alg:matroids-column} needs to fulfill. The framework of CRSs was formalized in~\cite{crspaper2011} and has since found many applications. A CRS is an algorithm that accepts a random subset $R \subseteq E$, that is not necessarily an independent set, and trims it down to an independent set $\pi(R) \subseteq R$ ($\pi(R) \in \calI$) such that, intuitively, each element is kept with a good probability. The study of CRSs has been primarily centered around product distributions; in other words, the algorithms usually require that each element $e \in E$ is included in $R$ independently with probability $x_e$, where $\bx \in \calP(M)$, i.e., the probability vector $\bx$ lies in the matroid polytope.

In our case, the sampling procedures implemented in Algorithms~\ref{alg:matroids-uniform} and~\ref{alg:matroids-column} do not satisfy the latter sampling independence, but a CRS can still be designed for them. 
Recently, Dughmi \cite{dughmi2020, dughmi2022} studied the existence of CRSs for non-product sampling procedures
and proved that this question is tightly connected with the Matroid Secretary Conjecture.
Although recent work by Qiu and Singla \cite{singla_submodular_dominance} indicate that standard CRS algorithms
can probably work for our sampling procedures, for the sake of simplicity and completeness, we describe
a simple (but slightly sub-optimal) CRS that works with any sampling strategy that falls under Definition~\ref{def:feasible-sampling}.

Finally, we remark that our use of CRSs is ``unconventional''. In our analysis, we need to condition on 
``outlier'' events of the sampling procedure and then apply the properties of a CRS. In order to do so,
we prove that even after conditioning on these events, the sampling procedure is still a feasible sampling strategy (Definition~\ref{def:feasible-sampling}) and thus the properties of our CRS still hold.

\begin{restatable}{definition}{deffeasiblesampling}[Feasible Sampling Strategy]\label{def:feasible-sampling}
    Let $M = (E, \calI)$ be a matroid. A randomized algorithm that outputs a set $R \subseteq E$,
    which includes each element $e \in E$ with marginal probability $\Pr[e \in R] = p_e$,
    is called a feasible sampling strategy if the following two properties hold.
    \begin{enumerate}
        \item The probability vector $p$, down-scaled by $\nicefrac{1}{2}$, must be in the matroid polytope,
        that is $p/2 \in \calP(M)$.
        \item For every element $e$ it should hold that
        $$
        \Pr[e \in \spann(R\setminus \{e\}) | e \in R] \leq \Pr[e \in \spann(R \setminus \{e\})].
        $$
    \end{enumerate}
\end{restatable}

\noindent The main property of our CRS is described in the theorem that follows. We defer its proof to Appendix~\ref{app:CR}.

\begin{restatable}{theorem}{crsmainthm}\label{thm:CRS-main}
     Let $M=(E, \calI)$ be a matroid, and $p \in P(M)$. Suppose that $R \subseteq E$ is a random sample
     generated by a feasible sampling strategy. Then $\textrm{Contention-Resolution}(R)$ outputs a set $\pi(R) \subseteq R$, with $\pi(R) \in \calI$, such that for each $e \in E$:
     $$\Pr[e \in \pi(R) | e \in R] \geq \frac{1}{8}\ .$$
\end{restatable}

\subsection{Column Decomposition} \label{subsec:column-decomposition}
In this section, we describe the partitioning of our prefix into groups of elements, which we call 
columns, with specific properties.
We also introduce the notation that will be used throughout the analysis in Section~\ref{subsec:matroids-analysis}.

\vspace{1em}
Suppose that our prefix consists of the first $\ell$ bases of the ordering.
We start by ordering the elements of $B_{\ell+1}$ in decreasing order of their
activation probability and separating them into two groups: elements with probability
at least $\nicefrac{10}{\ell}$ and elements with probability at most $\nicefrac{10}{\ell}$.
Formally, let $B_{\ell+1} = \{a_1, \dots, a_{r_1}\} \cup \{b_1, \dots, b_{r_2}\}$ be a partition
of $B_{\ell+1}$, such that:

\begin{enumerate}
    \item $p_{a_1} \geq p_{a_2} \geq \dots \geq p_{a_{r_1}} \geq p_{b_1} \geq \dots \geq p_{b_{r_2}}$
    \item $\forall i \in [r_1]: p_{a_i} \geq \nicefrac{10}{\ell}$
    \item $\forall i \in [r_2]: p_{b_i} < \nicefrac{10}{\ell}$.
\end{enumerate}

Furthermore, for every base $B_i$ of the $\ell$ bases of our prefix, we create a bijection $\pi_i: 
B_{\ell+1} \to B_i$, such that for every element $e \in B_{\ell+1}$, $(B_i \setminus \{\pi_i(e)\})\cup\{e\}$
is a base of the given matroid. These bijections are guaranteed to exist (\cref{thm:bij}) and can 
be easily computed in polynomial time. 
Notice that the elements of $B_{\ell+1}$ are always mapped to elements with higher activation 
probability, that is
$$\forall i \in [\ell],\ \forall e \in B_{\ell+1}:\ p_{\pi_i(e)} \geq p_{e}.$$
If the latter was not true, then there would exist an $i \in [\ell]$ and an element $e \in B_{\ell+1}$
such that $p_{\pi_i(e)} < p_{e}$. By the exchange property of our bijections, this would imply
that $p(B_i) < p((B_i \setminus \{\pi_i(e)\})\cup \{e\})$ which is a contradiction since
the base $(B_i \setminus \{\pi_i(e)\})\cup \{e\}$ should have been picked, instead of $B_i$,
to be the base with the $i$-th highest expected value in the ordering of the algorithm.

\vspace{1em}
\noindent We now create an $\ell \times r$ matrix by putting the elements of the $i$-th
base, $B_i$, on the $i$-th row. Inside each row, in the $j$-th column we put the
element that is the image, under $\pi_i$, of the $j$-th element of $B_{\ell+1}$ according to the order
described above. The matrix can be seen in \cref{fig:column-decomposition}.

\begin{table}[H]
\centering
\begin{tabular}{l|ccccccc}
$B_{1}$   & $\pi_1(a_1)$ & $\pi_1(a_2)$ & $\cdots$ & $\pi_1(a_{r_1})$ & $\pi_1(b_1)$ & $\cdots$ & $\pi_1(b_{r_2})$ \\
$B_{2}$   & $\pi_2(a_1)$ & $\pi_2(a_2)$ & $\cdots$ & $\pi_2(a_{r_1})$ & $\pi_2(b_1)$ & $\cdots$ & $\pi_2(b_{r_2})$ \\
$\vdots$  &              &              &          &                  &              &          &                  \\
$B_{\ell}$   & $\pi_\ell(a_1)$ & $\pi_\ell(a_2)$ & $\cdots$ & $\pi_\ell(a_{r_1})$ & $\pi_\ell(b_1)$ & $\cdots$ & $\pi_\ell(b_{r_2})$ \\ \hline
$B_{\ell+1}$ & $a_1$        & $a_2$        & $\cdots$ & $a_{r_1}$        & $b_1$        & $\cdots$ & $b_{r_2}$       
\end{tabular}
\caption{Column Decomposition of the bases $B_1, \dots, B_\ell$.}\label{fig:column-decomposition}
\end{table}

\noindent Throughout the remainder of the section we are using the following notation:

\begin{itemize}
    \item $C^1_j$ corresponds to the column above the element $a_j$
    and $C^2_j$ corresponds to the column above the element $b_j$.
    Formally,
    \begin{align*}
        \forall j \in [r_1]&:\ C^1_j = \{\pi_1(a_j), \pi_2(a_j), \dots, \pi_\ell(a_j)\}\\
        \forall j \in [r_2]&:\ C^2_j = \{\pi_1(b_j), \pi_2(b_j), \dots, \pi_\ell(b_j)\}\\
    \end{align*}
    
    \item We use $H$ to denote the elements of the prefix belonging to a column $C^1_j$ and $L$ to denote the elements of columns $C^2_j$. In other words, $H = \cup_{j \in [r_1]} C^1_j$ and $L = \cup_{j \in [r_2]} C^2_j$.
    \end{itemize}

\subsection{Analysis}\label{subsec:matroids-analysis}

In order to prove that Algorithm~\ref{alg:matroids} is a constant-factor approximation, 
it suffices to prove that there exists one prefix of the created ordering, 
for which one of the two created portfolios has value $\Theta(1)\cdot \OPT$. 
Let $B_1, \dots, B_n$ be the bases in the 
order created by the algorithm.
Our analysis will focus on the prefix consisting of the first $d$ bases, where $d$ is the smallest index 
such that
\begin{align}
\E_{A\sim \calD}\left[ \rank \left( \left( \cup_{i \in [d]} B_i\right) \cap A \right) \right] \geq \frac{\OPT}{2}. \tag{$\otimes$} \label{eq:prefix-rank-opt}
\end{align}

An immediate consequence of the selection of this prefix is that the optimal portfolio, 
$\calO = \{O_1, \dots, O_k\}$, 
must receive at least a constant-factor of its value from elements outside of the selected prefix. 
We formally present this in the following lemma where we use $E_d$ to denote the elements of the first $d$ bases and $E_r$ denote the rest of the elements, i.e. $E_d 
= \cup_{i \in [d]} B_i$ and $E_r = E \setminus E_d$.

\begin{lemma}\label{lem:matroids-opt-outside-prefix}
    $\E_{A\sim \calD}\left[\max_{O_i \in \calO} \abs{O_i \cap E_r \cap A}\right] \geq \OPT / 4.$
\end{lemma}

\begin{proof}[Proof of Lemma~\ref{lem:matroids-opt-outside-prefix}]
    We begin by proving that the expected rank of the active elements in the prefix is at most $3 \OPT / 4$. We can do
    so in the following way

    \begin{align*}
    \E_{A\sim \calD}\left[ \rank\left( E_d \cap A \right) \right] &= 
    \E_{A\sim \calD}\left[ \rank\left( \left(  \bigcup_{i=1}^d B_i \right) \cap A \right) \right]\\
    &\leq\E_{A\sim \calD}\left[ \rank\left( \left(  \bigcup_{i=1}^{d-1} B_i \right) \cap A \right)
    +  \rank\left( B_d \cap A \right)
    \right]\\
    &\leq  \frac{\OPT}{2} +  \frac{\OPT}{4} =  \frac{3\OPT}{4}
\end{align*}

\noindent where in the first inequality we used the subadditivity of the $\mathrm{rank}$ function and in the last we used~\ref{assum:maximum-probability-base-is-at-most-OPT/4} and the property of our prefix (Eq. (\ref{eq:prefix-rank-opt})).

\vspace{1em}
\noindent Having this result, the proof of the lemma is immediate since the value of the optimal portfolio $\calO$ restricted to the prefix $E_d$ can be at most the expected rank of the active 
elements in the prefix. We formally conclude the proof as follows:

\setcounter{equation}{0}
\begin{align}
    \OPT &=  \E_{A\sim\calD}\left[\max_{O_i \in \calO}\ \rank (O_i \cap A)\right]\\
    &\leq \E_{A \sim \calD} \left[\max_{O_i \in \calO}\ \rank (O_i \cap E_d \cap A)\right]
    + \E_{A \sim \calD} \left[\max_{O_i \in \calO}\ \rank (O_i \cap (E \setminus E_d) \cap A)\right]\\
    & \leq \E_{A\sim \calD}\left[ \rank\left( E_d \cap A \right) \right] + \E_{A \sim \calD} \left[\max_{O_i \in \calO}\ \rank (O_i \cap (E \setminus E_d) \cap A)\right]\\
    & \leq \frac{3\OPT}{4} + \E_{A \sim \calD} \left[\max_{O_i \in \calO}\ \rank (O_i \cap (E \setminus E_d) \cap A)\right],
\end{align}
where to get (2) we used the subaddditivity of the $\max$ operator and the $\rank$ function, 
to get (3) we used the monotonicity of $\rank$ and to get (4) we used our upper bound for the
expected rank of the active elements in $E_d$.
\end{proof}

\vspace{1em}
Next, we will prove that if we had access to $k$ independent copies of the highest expectation
base outside of our prefix, that is $B_{d+1}$, we could construct a portfolio with value
at least $\OPT/4$. We present this formally in the lemma bellow, which is the equivalent
of Lemma~\ref{lem:uniform-independent-M+1} for our analysis on uniform matroids.

\begin{lemma}\label{lem:matroids-b-d+1-opt}
    Let $B_{d+1} = \{e_1, \dots, e_r\}$ and $X_1, \dots, X_k$ be independent Poisson
    Binomial random variables following $\PB(p_{e_1}, \dots, p_{e_r})$.
    Then it holds that
    $$
    \E \left[ \max_{i \in [k]} X_i \right] \geq \frac{\OPT}{4}.
    $$
\end{lemma}

\begin{proof}[Proof of Lemma~\ref{lem:matroids-b-d+1-opt}]
    From Lemma~\ref{lem:matroids-opt-outside-prefix} we know that there exist $k$
    dependent Poisson Binomial random variables, namely the random variables 
    $|O_1 \cap E_r \cap A|, \dots , |O_k \cap E_r \cap A|$ for $A \sim \calD$,
    whose expected maximum is at least $\OPT/4$. We will show that we can transform these into the
    desired Poisson Binomials without decreasing their expected maximum.
    
    Recall that $B_{d+1}$ is the independent set with the highest expectation among all
    independent sets $I \subseteq E_r $. Due to \ref{assum:M-has-ell-disjoint-bases},
    without loss of generality, we can treat every independent set as a base that has been augmented
    with zero probability elements. Similarly to the column decomposition (\cref{subsec:column-decomposition}) where we created bijections from $B_{d+1}$ to bases in the prefix, we will create new bijections $\pi'$ from bases outside of the prefix to $B_{d+1}$. More formally, from~\cref{thm:bij}, 
    we know that for every base $I \subseteq 
    E_r$, there exists a bijection $\pi'$ from $I$ to $B_{d+1}$ that maps every element
    $e \in I$ to a distinct element $\pi'(e) \in B_{d+1}$ such that $(B_{d+1} \setminus \{\pi'(e)\})\cup 
    \{e\}$ is an independent set. 
    This directly implies that $p_{e} \leq p_{\pi'(e)}$ since
    the opposite would mean that $p((B_{d+1} \setminus \{\pi'(e)\})\cup\{e\}) > p(B_{d+1})$ which contradicts the fact that $B_{d+1}$ has the highest expectation among bases in $E_r$.

    \vspace{1em}
    \noindent We are now ready to describe the transformation of the random variables $|O_i \cap E_r \cap A|$ to the desired Poisson Binomials. We do so in the following two steps:

    \begin{enumerate}
        \item For every element $e$ that appears in more than one of the sets $(O_i \cap E_r )$, we replace it with a new, independent, Bernoulli random variable with activation probability $p_e$.
        \item For every set $(O_i \cap E_r )$, we use the bijection $\pi'_i: (O_i \cap E_r ) \to B_{d+1}$ given by \cref{thm:bij}, and we increase the
        probability of every element $e \in (O_i \cap E_r)$ to $p_{\pi'_i(e)}$.
    \end{enumerate}

    Both transformations do not decrease the expected maximum of the Poisson Binomials, since
    increasing the trial probabilities (Lemma~\ref{lem:binomials-bigger-p-better})
    and making the random variables independent (Lemma~\ref{lem:binomials-independent-better}) can only
    increase the expectation of their maximum.
\end{proof}

\vspace{1em}
\paragraph{Overview of the analysis.} As in our analysis for the uniform matroid (\cref{sec: Uniform Matroids}), we will try to argue that, after an appropriate conditioning on the randomness
of the instance, the sampled solutions behave as independent copies of $B_{d+1}$.
We remark that due to the structure of our prefix, for every element $e \in B_{d+1}$,
we have $d$ unique elements with higher activation probability in the prefix. 
Notice that if the element $e \in 
B_{d+1}$ has activation probability $p_e \geq \nicefrac{10}{d}$, we expect
that at least one of its $d$ representative elements in the prefix will be active. 
More precisely, with constant probability, at least $\nicefrac{p_e \cdot d}{2}$ elements from
its column will be active. In addition, the column portfolio will sample one of these active elements with  probability $\nicefrac{p_e}{2}$, since it includes one uniformly random element from each column.
Therefore, on a high level, our proof strategy for the elements $e \in B_{d+1}$ with $p_e \geq \nicefrac{10}
{d}$ will be to argue that the column portfolio samples an active element from their column
with probability roughly $p_e$. Critically, after conditioning on the event that a column
has ``enough'' active elements, the values of the $k$ samples of the column portfolio
from this column become independent random variables. This will allow us to compare
the produced value to the value of independent copies of $B_{d+1}$ which is
$\Theta(1) \cdot \OPT$ (Lemma~\ref{lem:matroids-b-d+1-opt}).


On the other extreme, this ``local'' analysis will fail for elements $e \in B_{d+1}$ with $p_e < 
\nicefrac{10}{d}$ because it is unlikely that we will see enough active elements in their column. 
For these elements, we show that the uniform portfolio succeeds in approximating the value of $B_{d+1}$.
To do this, we will focus on the biggest independent set of active elements in the prefix
and we will try to lower bound its intersection with our randomly sampled solutions.

Our analysis will, therefore, distinguish the above two cases. That is, whether the expectation
of the maximum of $k$ copies of $B_{d+1}$ is dominated by the ``high'' or the ``low'' probability
elements. We will then show that the column portfolio and the uniform portfolio
achieve near-optimal values in the two respective cases.

\vspace{1em}
\noindent More formally, let us define the following independent random variables
\begin{align*}
X_1, \dots, X_k &\sim \PB(\{p_e : e \in B_{d+1} \text{ s.t. } p_e \geq \nicefrac{10}{d}\}),\\
Y_1, \dots, Y_k &\sim \PB(\{p_e : e \in B_{d+1} \text{ s.t. } p_e < \nicefrac{10}{d}\}).
\end{align*}

\noindent We will distinguish the following cases:


\newlist{Cases}{enumerate}{2}
\setlist[Cases]{label=(Case \arabic*), font=\textbf, itemindent=0.78cm}

\begin{Cases}\label{matroid-cases}
        \item $\E\left[\max_{i \in [k]} X_i \right] \geq \E\left[\max_{i \in [k]} Y_i \right]$.
        \\ Due to Lemma~\ref{lem:matroids-b-d+1-opt}, this condition implies that 
        $\E\left[\max_{i \in [k]} X_i \right] \geq \OPT / 8$.\label{case1:definition}
        
        \item $\E\left[\max_{i \in [k]} X_i \right] < \E\left[\max_{i \in [k]} Y_i \right]$.
        \\ Due to Lemma~\ref{lem:matroids-b-d+1-opt}, this condition implies that 
        $\E\left[\max_{i \in [k]} Y_i \right] \geq \OPT / 8$.\\
        In addition, since all the activation probabilities of the trials of the random variabels $X_1, \dots, X_k$ are bigger than
        those of the random variables $Y_1, \dots, Y_k$, but the expected maximum of the latter is larger, it must hold that
        the random variables $Y_i$ have more trials, that is $r_2 > r_1$ (Lemma~\ref{lem:PB-domination-more-trials}).
        \label{case2:definition}
\end{Cases}

\vspace{1em}
\noindent In Sections \ref{subsec:matroids-case-1} and \ref{subsec:matroids-case-2},
we focus on each case separately and prove that the column portfolio (Algorithm~\ref{alg:matroids-column}) and the uniform portfolio (Algorithm~\ref{alg:matroids-uniform}) are constant-factor
approximations in the corresponding case.

\subsubsection{Case 1: ``High'' probability items} \label{subsec:matroids-case-1}

In this section, we prove that the column portfolio (Algorithm~\ref{alg:matroids-column}) is a constant
factor approximation under~\ref{case1:definition}. We will denote the random sets sampled by the algorithm as $V_1, \dots, V_k$ and their respective independent sets, after running the contention resolution scheme, as $S_1, \dots, S_k$.

To lower bound the value of the column portfolio, $\E_{\ALG, A\sim \calD}\left[ \max_{i \in [k]} |S_i \cap A| \right]$, by $\Theta (1)\cdot \OPT$, we first define a product distribution $\Tilde{\calD}$ so that in every column all elements have the same activation probability as their representative element in $B_{d+1}$. Then, our analysis follows the three steps described below.
\newlist{Steps}{enumerate}{2}
\setlist[Steps]{label=(Step \arabic*), font=\textbf, itemindent=2.3cm}

\begin{Steps}
    \item $\E_{\ALG, A\sim \calD}\left[ \max_{i \in [k]} |S_i \cap A| \right] \geq \E_{\ALG, \widetilde{A}\sim \widetilde{\calD}}\left[ \max_{i \in [k]} |S_i \cap \widetilde{A}| \right].$\label{step: first}
    \item $\E_{\ALG, \widetilde{A}\sim \widetilde{\calD}}\left[ \max_{i \in [k]} |S_i \cap \widetilde{A}| \right] \geq \frac{1}{8} \cdot 
    \E_{\ALG, \widetilde{A}\sim \widetilde{\calD}}\left[ \max_{i \in [k]} |V_i \cap \widetilde{A}|  \right].$\label{step: second}
    \item $\E_{\ALG, \widetilde{A}\sim \widetilde{\calD}}\left[ \max_{i \in [k]} |V_i \cap \widetilde{A}|  \right] \geq \Theta (1) \cdot \OPT.$\label{step: third}
\end{Steps}

\newcommand{\tcD}{\widetilde{\calD}}
\newcommand{\tA}{\widetilde{A}}

We begin by formally defining the new product distribution $\tcD$. Let $\tA$ be a random variable drawn from $\tcD$, then
\begin{enumerate}
    \item $\Pr [e \in \Tilde{A}] = \Pr [e \in A], \forall e \in E \setminus \left(\cup_{i \in [d]} B_i \right)$,
    \item $ \Pr [ e \in \Tilde{A}] = p_{a_j},  \forall e \in C^1_{j}, \forall j \in [r_1]$, and
    \item $ \Pr [ e \in \Tilde{A}] = 0,  \forall e \in C^2_{j}, \forall j \in [r_2]$.
\end{enumerate}

Since $\Pr [e \in A] \geq \Pr [e \in \Tilde{A}] $ $\forall e \in E$ and both $\calD, \tcD$ are product distributions, there exists trivial coupling of the random variables $A$ and $\tA$ such that event $\{ e \in \Tilde{A} \}$ implies event $\{ e \in A\}$. Using that coupling, \ref{step: first} trivially holds.

\vspace{1em}
\noindent We now proceed in proving~\ref{step: second}.

\begin{lemma}\label{lem:matroids-case1-s-is-theta-1-v}
    $\E_{\ALG, \tA\sim \tcD}\left[ \max_{i \in [k]} |S_i \cap \tA| \right] \geq \frac{1}{8} \cdot 
    \E_{\ALG, \tA\sim \tcD}\left[ \max_{i \in [k]} |V_i \cap \tA|  \right].$
\end{lemma}
\begin{proof}[Proof of Lemma~\ref{lem:matroids-case1-s-is-theta-1-v}]

\newcommand{\is}{{i^{*}}}
\newcommand{\Cs}{C^{*}}
\

\noindent We start by defining two helpful random variables: $\is = \argmax_{i \in [k]} |V_i \cap \tA|$ denotes the index of the set, among sets $V_i$, with the most active elements and $\Cs = \{j \in [r_1]: \abs{V_\is \cap \tA \cap C^1_j} > 0 \}$ contains all indexes of the columns where $V_\is$ has sampled active elements.

\vspace{1em}
\noindent We argue that even after conditioning on some values for $\is$
and $\Cs$, the sets $V_i$ have been produced by a feasible sampling strategy,
as defined in \cref{def:feasible-sampling}, and therefore we can use the core property of our contention resolution scheme (\cref{thm:CRS-main}).

\vspace{1em}
Conditioning on $\is$ and $\Cs$, gives us the information of whether 
our algorithm sampled an active or inactive element in a column, but, crucially, it 
doesn't tell us which element was sampled. 
To see this, we describe an equivalent stochastic process which generates the values $\abs{V_i \cap \tA}$ for all $i$. First, for $j \in  [r_1]$ draw the independent random variables $A_j \sim Bin(d, p_{a_j})$. Since any  two columns are disjoint, the number of active elements in each column $C^1_j$, according to $\widetilde{D}$, follows the same distribution as random variables $A_j$. In the second step, for each instantiation of random variables $A_j$ to numbers $\alpha_j$, draw $k$ independent Poisson binomial random variables with probabilities $\nicefrac{\alpha_1}{d} , \dots, \nicefrac{\alpha_{r_1}}{d}$. The latter random variables follow the same distribution as $\abs{V_1 \cap \tA}, \dots, \abs{V_{r_1} \cap \tA}$. At this moment, the values of $\is$ and $\Cs$ have been instantiated but the identity of which element is active was not. Furthermore, all elements in a column are equally likely to be active. Therefore, for any $i' \in [k]$ and any $C' \subseteq [r_1]$, we can conclude the following hold for the sampling process of the sets $V_i$ conditioned on the event $\is = i'$ and $\Cs = C'$:
\begin{enumerate}
    \item The vector of marginal probabilities, with which each element is included in 
    a sampled set, lies in the matroid polytope. This is true because 
    for every $e \in H \cup L$ and every $i \in [k]:\ \Pr_{\tA \sim \tcD, 
    \ALG}[e \in V_i | \is = i', \Cs = C'] = \nicefrac{1}{d}$. As a result, the vector 
    of marginal probabilities can be written as the uniform combination of the 
    indicator vectors of the $d$ bases in our prefix.

    \item Sampling an element $e$ doesn't make it more likely for $e$ to be spanned
    from the sampled elements of the other columns. Formally, for all $e \in H \cup L$
    and all $i \in [k]$,
    \begin{align*}
    \Pr_{\tA \sim \tcD, \ALG}\left[ e \in \spann(V_i \setminus e) \Big| e\in V_i, \is = i', \Cs = C' \right] \leq \\ \leq \Pr_{\tA \sim \tcD, \ALG}\left[ e \in \spann(V_i \setminus e) \Big| \is = i', \Cs = C' \right].
    \end{align*}
    The latter can be proven by the principle of deferred decisions after fixing the outcome of all columns other than $e$'s column and noticing that sampling $e$
    only makes it less likely for it to be in $\spann(V_i \setminus e)$.
\end{enumerate}

\noindent Due to the above points, \cref{def:feasible-sampling} is fulfilled. Therefore, from 
\cref{thm:CRS-main} we have that $\forall i', i \in [k],\ \forall C' \subseteq [r_1]$
the following holds

\begin{align}
     \Pr_{\tA \sim \tcD, \ALG}\left[ e \in S_i \Big| \is = i', \Cs = C' \right] \geq \frac{1}{8} \Pr_{\tA \sim \tcD, \ALG}\left[ e \in V_i \Big| \is = i', \Cs = C' \right] \notag.
\end{align}

\noindent Since the above holds for every $C' \subseteq [r_1]$, by the
law of total probability we get that $\forall i, i' \in [k]$

\begin{align}
    \Pr_{\tA \sim \tcD, \ALG}\left[ e \in S_i \Big| \is = i' \right] \geq \frac{1}{8} \Pr_{\tA \sim \tcD, \ALG}\left[ e \in V_i \Big| \is = i'\right] \tag{$\wr$}\label{eq:conditioned-crs}.
\end{align}

\noindent The above equation, on a high level, tells us that each $S_i$ will contain, on expectation, $1/8$ of the elements
of the corresponding set $V_i$. Therefore, the proof of the lemma can formally be concluded as follows.

\setcounter{equation}{0}
\begin{spreadlines}{2ex}
\begin{align}
    &\E_{\tA \sim \tcD, \ALG} \left[\max_{i \in [k]} \left|S_i \cap \tA \right|\right]\\
    &\geq \E_{\tA \sim \tcD, \ALG} \left[\left|S_\is \cap \tA \right|\right]\\
    & = \sum_{i' \in [k]} \E_{\tA \sim \tcD, \ALG} \left[\left|S_\is \cap \tA \right| \Big| \is = i'\right]
    \Pr_{\tA \sim \tcD, \ALG} \left[ \is = i'\right]\\
    & = \sum_{i' \in [k]} \sum_{e \in H\cup L}\Pr_{\tA \sim \tcD, \ALG} \left[e \in (S_\is \cap \tA) \Big| \is = i'\right]
    \Pr_{\tA \sim \tcD, \ALG} \left[ \is = i'\right]\\
    & = \sum_{i' \in [k]} \sum_{e \in H\cup L}\Pr_{\tA \sim \tcD, \ALG} \left[e \in S_\is  \Big| \is = i'\right] \Pr_{\tA \sim \tcD, \ALG} \left[e \in \tA \Big| \is = i'\right] \Pr_{\tA \sim \tcD, \ALG} \left[ \is = i'\right]\\
    & \geq \frac{1}{8} \sum_{i' \in [k]} \sum_{e \in H\cup L}\Pr_{\tA \sim \tcD, \ALG} \left[e \in V_\is  \Big| \is = i'\right] \Pr_{\tA \sim \tcD, \ALG} \left[e \in \tA \Big| \is = i'\right] \Pr_{\tA \sim \tcD, \ALG} \left[ \is = i'\right]\\
    & = \frac{1}{8} \E_{\tA \sim \tcD, \ALG} \left[\left|V_\is \cap \tA \right|\right]\\
    & = \frac{1}{8} \E_{\tA \sim \tcD, \ALG} \left[\max_{i \in [k]} \left|V_i \cap \tA \right|\right]
\end{align}
\end{spreadlines}

\noindent where to get (4) we used the linearity of expectation,
to get (5) we used the fact that the events $e \in S_\is$ and $e \in \tA$
are independent under the condition of $\is = i'$, to get (6) we used Equation~(\ref{eq:conditioned-crs})
and to get (8) we used the definition of $\is$.
\end{proof}

\vspace{1em}
\noindent We now move on to the final step of the analysis, \ref{step: third}, where we will
prove that the expected value of the sampled sets $V_1, \dots, V_k$ is $\Theta(1) \cdot \OPT$.
We restate this step in the lemma below.

\begin{lemma}\label{lem:matroids-case-1-vi-are-opt}
    $\E_{\ALG, \widetilde{A}\sim \widetilde{\calD}}\left[ \max_{i \in [k]} |V_i \cap \widetilde{A}|  \right] \geq \frac{7}{160} \OPT .$
\end{lemma}

\begin{proof}[Proof of Lemma~\ref{lem:matroids-case-1-vi-are-opt}]

    In order to prove this lemma, we will focus on the value that our portfolio produces
    from the columns in $H$, i.e. from the columns that correspond to the elements of $B_{d+1}$
    with activation probability at least $\nicefrac{10}{d}$. Of course, it holds that
    $$
    \E_{\ALG, \widetilde{A}\sim \widetilde{\calD}}\left[ \max_{i \in [k]} |V_i \cap \widetilde{A}|  \right] \geq \E_{\ALG, \widetilde{A}\sim \widetilde{\calD}}\left[ \max_{i \in [k]} |V_i \cap H \cap \widetilde{A}|  \right].
    $$
    On a high level, the plan of our proof is the following. For every column $C^1_j$, we define
    an event that captures whether this column has ``enough'' active elements after sampling 
    an instance $\tA \sim \tcD$. This event happens with constant probability, independently, for every column
    and guarantees that a random sample from the corresponding column, let it be $C^1_j$, will be an 
    active element with probability roughly $p_{a_j}$. This viewpoint allows us to lower bound the
    value of our portfolio by a stochastic process that first samples a Bernoulli random variable
    for each column, to decide whether it had ``enough'' active elements or not, and then samples
    some Bernoulli random variables that correspond to the event of whether the algorithm  
    picked an active element or not. We then prove that the value produced by this process
    is comparable to the value of $k$ independent copies of the elements $\{a_1, \dots, a_{r_1}\}$,
    which is $\Theta(1) \cdot \OPT$ due to the definition of \ref{case1:definition}.

    \vspace{1em}
    More formally, for every one of the columns $C^1_1, \dots, C^1_{r_1}$, we define an indicator
    variable $G_i$ as
$$
\forall i \in [r_1]:\ G_i = \scalebox{1.3}{$\indic$} \left\{ \abs{C^1_i \cap \tA} \geq \frac{d \cdot p_{a_i}}{2} \right\}.
$$

\noindent We will lower bound the probability that $G_i = 1$ by a constant for every column $i \in [r_1]$. To do so, we first lower bound the expected number of active elements in a column $C^1_i$,
for all $i \in [r_1]$, as follows
$$\E_{\tA \sim \tcD} [\abs{C^1_i \cap \tA}] = \sum_{e \in C^1_i} p_e = \abs{C^1_i} \cdot p_{a_i} = d \cdot p_{a_i} \geq d \cdot (10/d) = 10.$$ 

\noindent Now, we can lower bound the probability that the event $G_i$ happens in the following way

\setcounter{equation}{0}
\begin{align}
        \E_{\tA \sim \tcD} [G_i] &= \Pr_{\tA \sim \tcD} \left[ \abs{C^1_i \cap \tA} \geq \frac{d \cdot p_{a_i}}{2} \right]\\
    &=  \Pr_{\tA \sim \tcD} \left[ \abs{C^1_i \cap \tA}\geq \E_{_{\tA \sim \tcD}} [\abs{C^1_i \cap \tA}] /2 \right]\\
    &= 1 - \Pr_{\tA \sim \tcD}\left[ \abs{C^1_i \cap \tA} < (1-1/2) \cdot \E_{\tA \sim \tcD} [\abs{C^1_i \cap \tA}] \right]\\
    &\geq 1 - e^{ - \frac{\frac{1}{4} \cdot  \E_{\tA \sim \tcD} [\abs{C^1_i \cap \tA}] }{2} }\\
     &\geq 1 - e^{ - \frac{10}{8} }\\
     &>7/10
\end{align}

\noindent where for (3) we used that $\E_{\tA \sim \tcD} [\abs{C^1_i \cap \tA}] = d \cdot p_{a_i} $,
to get (4) we used a Chernoff bound (\cref{cor:chernoff}) and to get (5) we used that $\E_{\tA \sim \tcD} [\abs{C^1_i \cap \tA}] \geq 10$.

\vspace{1em}
Notice that when the event $G_i$ happens for column $C^1_i$, the probability to pick
an active element after sampling one element uniformly at random is at least $\frac{p_{a_i} d}{2} 
\cdot \frac{1}{d} \geq \frac{p_{a_i}}{2}$. Therefore, for every $j \in [k]$, it holds that

$$
\Pr_{\tA \sim \tcD, \ALG}\left [ |V_j \cap C^1_i \cap \tA| = 1\right] \geq
\Pr_{\tA \sim \tcD} \left[ G_i = 1\right] \cdot \Pr_{\tA \sim \tcD, \ALG} \left [ |V_j \cap C^1_i \cap \tA| = 1 \Big| G_i = 1\right] \geq \frac{7 p_{a_i}}{20}.
$$

\noindent Furthermore, after conditioning on the number of active elements for the column $C^1_i$, the 
values of the random variables $|V_1 \cap C^1_i \cap \tA|, \dots , |V_k \cap C^1_i \cap \tA|$
depend only on the internal randomness of the algorithm and are, thus, independent.

\vspace{0.5em} We will now define the afforementioned stochatic process as follows. Let $Z_1, \dots, Z_{r_1}$ be i.i.d Bernoulli random variables that follow $\Be(0.7)$. Let also $X_1^{(1)}, X_2^{(1)},
\dots, X_{r_1}^{(k)}$ be independent Bernoulli random variables such that $\forall i \in [r_1]$
and $\forall j \in [k]$, $X_i^{(j)} \sim \Be(\nicefrac{ p_{a_i}}{2})$. Due to the above,
we know that the value of $j$-th set sampled by the algorithm, i.e. $|V_j \cap H \cap \tA|$,
stochastically dominates $\sum_{i \in [r_1]} Z_i \cdot X_i^{(j)}$. Therefore, we get that

$$
\E_{\ALG, \widetilde{A}\sim \widetilde{\calD}}\left[ \max_{i \in [k]} |V_i \cap H \cap \widetilde{A}|  \right] \geq \E \left[ \max_{i \in [k]} \sum_{j \in [r_1]} Z_j \cdot X_j^{(i)}\right]
$$

\noindent We finish the proof by defining the i.i.d Poisson Binomial random variables $Y_1, \dots, Y_k$ where $Y_i \sim \PB(\{p_{a_1}, \dots, p_{a_{r_1}}\})$. Then, it holds that
$$
\E \left[ \max_{i \in [k]} \sum_{j \in [r_1]} Z_j \cdot X_j^{(i)}\right] 
\geq \frac{7}{20} \E \left[ \max_{i \in [k] } Y_i \right] \geq \frac{7}{160} \OPT
$$
\noindent where for the first inequality we applied Lemma~\ref{lem:poisson-binomial-random-subset}
and for the second we used the definition of \ref{case1:definition}.

\end{proof}

\subsubsection{Case 2: ``Low'' probability items} \label{subsec:matroids-case-2}

\newcommand{\tV}{\widetilde{V}}
\newcommand{\Imax}{I^{\max}}

In this section we prove that the uniform portfolio (Algorithm~\ref{alg:matroids-uniform}) is
a constant-factor approximation of the optimal portfolio under \ref{case2:definition}.
Let $\tV_1, \dots, \tV_k$ be the multisets formed by the algorithm after sampling $r$ elements uniformly at 
random (with replacement) from the prefix, $V_1, \dots, V_k$ be the corresponding sets that contain the unique elements sampled, and $S_1, \dots, S_k$ be the independent
sets returned by the contention resolution scheme.

Our analysis will follow two steps: First, we show that in order to analyze the value of the uniform 
portfolio it suffices to prove that the sampled sets $V_1, \dots, V_k$ contain a large enough 
independent set of active elements before they are given to the contention resolution scheme.
Then, we prove that, on expectation, one of these sets contains an independent set of
active elements with size $\Theta(1) \cdot \OPT$. We formally describe these two steps below.

\begin{Steps}
    \item $\E_{\ALG, A\sim \calD}\left[ \max_{i \in [k]} |S_i \cap A| \right] \geq \frac{1}{8}
    \E_{\ALG, A\sim \calD}\left[ \max_{i \in [k]} |V_i \cap \Imax(A)| \right].$ \label{step-1-case-2}
    \item $\E_{A \sim \calD, \ALG}\left[ \max_{i \in [k]} |V_i \cap \Imax(A)| \right] \geq \Theta(1) \cdot \OPT.$\label{step-2-case-2}
\end{Steps}

\noindent where $\Imax(A)$ denotes a maximum cardinality independent set of active elements in the prefix.

\vspace{1em}
\noindent We begin by proving \ref{step-1-case-2}, in a slightly more general form, which is stated in the
following lemma.

\begin{lemma}\label{lem:matroids-case-2-s-is-theta-1-v}
    Let $I(A)$ be any independent set of the active elements in the prefix, i.e. $I(A) \subseteq (A\cap(H \cup L))$ 
    and $I(A) \in \calI$.
    It holds that
    $$
    \E_{\ALG, A\sim \calD}\left[ \max_{i \in [k]} |S_i \cap A| \right] \geq \frac{1}{8}
    \E_{\ALG, A\sim \calD}\left[ \max_{i \in [k]} |V_i \cap I(A)| \right].
    $$
\end{lemma}

\begin{proof}[Proof of Lemma~\ref{lem:matroids-case-2-s-is-theta-1-v}]
    Let $i^*$ be the random variable that denotes the maximizer of the expression $\max_{i \in [k]} |V_i \cap I(A)|$
    and $I^* = V_{i^*} \cap I(A)$. Suppose that we want to condition on the event that $I^* = S$ for some fixed
    independent set $S \subseteq (H \cup L)$. Then, the marginal probabilities with which each element is included in $V_{i^*}$, are as follows:
    \begin{itemize}
        \item $\forall e \in S:\ \Pr[e \in V_{i^*} | I^* = S] = 1$, as the conditioning implies that $S \subseteq V_{i^*}$.
        \item $\forall e \in (H\cup L)\setminus S:\ \Pr[e \in V_{i^*} | I^* = S] \leq 1-(1-\nicefrac{1}{dr})^r \leq 1-e^{-\nicefrac{1}{d}} \leq \nicefrac{1}{d}$, as
        $V_{i^*}$ was created by sampling $r$ elements uniformly at random from the $d\cdot r$ elements of the prefix.
    \end{itemize}

    Notice that even after the conditioning, $V_{i^*}$ has been sampled from a feasible sampling strategy 
    (Definition~\ref{def:feasible-sampling}), since for the vector $p$, with $p_e = \Pr[e \in V_{i^*} | I^* = S]$, 
    the down-scaled vector $p/2$ lies in the matroid polytope. The latter is immediate as one can write the vector $p/2$
    as the uniform combination of the indicator vector of the independent set $S$ and the vector $x$ with
    $x_e = \nicefrac{1}{d}$ for $e \in (H \cup L)$, which is in the matroid polytope because it can be written
    as the uniform combination of the indicator vectors of the $d$ bases of the prefix.
    Apart from that, conditioning on the event that an element $e$ was sampled by $V_{i^*}$
    doesn't make it more likely for it to be in $\spann(V_{i^*}\setminus\{e\})$.

\vspace{1em}
    Therefore, after obtaining $S_{i^*}$ from the contention resolution scheme on $V_{i^*}$, due to
    \cref{thm:CRS-main}, we get that for any independent set $S \in (H \cup L)$ 
    and any element $e \in (H \cup L)$, it holds that

    \begin{align}
    \Pr_{\ALG, A\sim \calD}\left[ e \in S_{i^*} \Big| I^* = S \right] \geq \frac{1}{8} \Pr_{\ALG, A\sim \calD}\left[ e \in V_{i^*} \Big| I^* = S \right]. \tag{$\oplus$} \label{eq:matroids-case-2-CRS}    
    \end{align}

    \noindent The proof of the lemma is practically over, since applying the above equation to the elements of $I^*$
    will give us that, on expectation, $1/8$ of the elements of $I^*$ will be retained by the contention resolution scheme.
    We formally perform this last step below.
    
    \setcounter{equation}{0}
    \begin{align}
        \E_{\ALG, A\sim \calD}\left[ \max_{i \in [k]} |S_i \cap A| \right] 
        &\geq \E_{\ALG, A\sim \calD}\left[ |S_{i^*} \cap A| \right]\\ 
        &= \sum_{\substack{S \subseteq (H\cup L) \\ \text{s.t. } S \in \calI}} \E_{\ALG, A\sim \calD}\left[ |S_{i^*} \cap A| \Big| I^* = S\right]
        \Pr_{\ALG, A\sim \calD} \left [ I^* = S \right]\\
        &\geq \sum_{\substack{S \subseteq (H\cup L) \\ \text{s.t. } S \in \calI}} \E_{\ALG, A\sim \calD}\left[|S_{i^*} \cap S| \Big| I^* = S\right] \Pr_{\ALG, A\sim \calD} \left [ I^* = S \right]\\
        & = \sum_{\substack{S \subseteq (H\cup L) \\ \text{s.t. } S \in \calI}} \sum_{e \in S} \Pr_{\ALG, A\sim \calD} \left [ e \in S_{i^*} \Big| I^* = S \right]  \Pr_{\ALG, A\sim \calD} \left [ I^* = S \right]\\
        &\geq \frac{1}{8} \sum_{\substack{S \subseteq (H\cup L) \\ \text{s.t. } S \in \calI}}  \sum_{e \in S} \Pr_{\ALG, A\sim \calD} \left [ e \in V_{i^*} \Big| I^* = S \right]  \Pr_{\ALG, A\sim \calD} \left [ I^* = S \right]\\
        & \geq \frac{1}{8} \sum_{\substack{S \subseteq (H\cup L) \\ \text{s.t. } S \in \calI}} |S| \cdot \Pr_{\ALG, A\sim \calD} \left [ I^* = S \right]\\
        & = \frac{1}{8} \E_{\ALG, A\sim \calD}\left[ |I^*| \right]\\
        & = \frac{1}{8} \E_{\ALG, A\sim \calD}\left[ |V_{i^*} \cap I(A)| \right]\\
        & = \frac{1}{8} \E_{\ALG, A\sim \calD}\left[ \max_{i \in [k]} |V_{i} \cap I(A)| \right],
    \end{align}

    \noindent where for (2) we used the law of total expectation, for (3) we used that $S \subseteq A$ due
    to the condition $I^* = S$, for (4) we used the linearity of expectation, for (5) we used
    Equation~(\ref{eq:matroids-case-2-CRS}), for (6) we used that $\Pr[e \in V_{i^*}| I^* = S] = 1$ for $e \in S$,
    for (8) we used the definition of $I^*$ and for (9) we used the definition of $V_{i^*}$.
\end{proof}

Having proved Lemma~\ref{lem:matroids-case-2-s-is-theta-1-v}, the analysis of the
produced portfolio reduces to proving that the sampled sets $V_i$ contain
a large enough independent set of active elements. To achieve this, after an appropriate
conditioning on the instance, we will focus on the largest independent set of active elements
in the prefix. Then, we will show that the probability that one sample of the algorithm fell
in this set is comparable to the average probability in $\{p_{b_1}, \dots, p_{b_{r_2}}\}$.
This observation will give us a way to relate the value of our sets to $k$
independent trials of $\PB(\{p_{b_1}, \dots, p_{b_{r_2}} \})$, which, due to the definition
of \ref{case2:definition}, is $\Theta(1) \cdot \OPT$.

We begin this analysis by proving a lower bound on the expected rank of active elements
in the prefix. As a reminder, we use $H$ to denote $\cup_{j \in [r_1]}C^1_j$
and $L$ to denote $\cup_{j \in [r_2]} C^2_j$. We will use $\mu$ to denote
$\sum_{j \in [r_2]} p_{b_j}$.

\begin{lemma}\label{lem:matroids-case2-unified-lower-bound-exp-rank-in-prefix}
    $\E_{A \sim \calD}\left[ \rank(A\cap(H\cup L))\right] \geq (1-\nicefrac{1}{e}) \mu d / 10$.
\end{lemma}

\begin{proof}[Proof of Lemma~\ref{lem:matroids-case2-unified-lower-bound-exp-rank-in-prefix}]
    Due to the property of the selected prefix (Eq. (\ref{eq:prefix-rank-opt})), it holds that
    $\E_{A \sim \calD}\left[ \rank(A\cap(H\cup L)\right] \geq \OPT / 2$. Therefore, if $\OPT > \mu d$
    then the lemma holds. We will prove that it also holds in the case where $\OPT \leq \mu d$.
    
\vspace{1em}
    \noindent Due to the monotonicity of the $\rank$ function, it suffices to focus on the rank of active elements in $L$, because $$\E_{A \sim \calD}\left[ \rank(A\cap(H\cup L))\right] \geq \E_{A \sim \calD}\left[ \rank(A\cap L)\right].$$ 
    
    \vspace{1em}
    \noindent To prove the desired inequality, we define a new product 
    distribution $\tcD$ that gives zero activation probability to the elements in $E\setminus L$
    and downscales the probabilities of the elements in $L$, so that all elements of a column $C^2_j$ have probability $p_{b_j}/10$.
    Formally, let $\tA \sim \tcD$, then
    \begin{enumerate}
        \item $\Pr_{\tA \sim \tcD}\left[e \in \tA \right] = 0,\ \forall e \in E\setminus L$
        \item $\Pr_{\tA \sim \tcD}\left[ e \in \tA \right] = p_{b_j}/10,\ \forall j\in [r_2] \text{ and } \forall e \in C^2_j$
    \end{enumerate}

    Since $\calD$ and $\tcD$ are both product distributions and $\Pr_{A \sim \calD}[e \in A] \geq 
    \Pr_{\tA \sim \tcD}[ e \in \tA] \ \forall e \in E$, it follows that
    $$
    \E_{A \sim \calD} \left[ \rank(A \cap L) \right ] \geq \E_{\tA \sim \tcD} \left [ \rank(\tA \cap L) \right ].
    $$

    \noindent In addition, the vector $\tilde{\bp} \in [0,1]^{|E|}$, with $\tilde{p}_e = 
    \Pr_{\tA \sim \tcD}[e \in \tA]$, lies in the matroid polytope. This is because $\tilde{\bp}$
    can be seen as a downscaled version of the uniform combination of the $d$ bases in the
    prefix, since $\forall e \in L:\ \tilde{p}_e \leq \nicefrac{1}{d}$ and $\forall e \not\in L:\ \tilde{p}_e = 0$. Therefore, the proof can be concluded as follows.
    \setcounter{equation}{0}
    \begin{spreadlines}{2ex}
    \begin{align}
    \E_{\tA \sim \tcD} \left [  \rank (\tA \cap L) \right ] 
    & \geq \E_{\tA \sim \tcD} \left [  \rank (\tA) \right ] \\
    & \geq (1-\nicefrac{1}{e}) \cdot \E_{\tA \sim \tcD} \left [  |\tA| \right ]\\
    & = (1-\nicefrac{1}{e}) \sum_{j \in [r_2]} \sum_{e \in C^2_j} \tilde{p}_e\\
    & = (1-\nicefrac{1}{e}) \frac{d}{10} \cdot \sum_{j \in [r_2]} p_{b_j}\\
    & = (1-\nicefrac{1}{e}) \mu d / 10
    \end{align}
    \end{spreadlines}

    \noindent where (1) holds because $\tA \subseteq L$ with probability $1$,
    (2) is the consequence of any $(1-\nicefrac{1}{e})$ contention resolution scheme on $\tilde{\bp}$ 
    (Corollary~\ref{cor: Expected rank(A) is at least (1-1/e) Expected size of A})
    and (4) follows from the definition of $\tcD$.
\end{proof}

Let $\Imax(A)$ be the random variable that denotes a maximum rank set in 
$A \cap ( H \cup L)$.
To analyze the value of the sampled sets, we will focus on the instances where $|\Imax(A)|$ is at 
least a constant fraction of its expectation.
We denote the set of these instances as $W$ and we formally define it as
$$
W = \left \{ A \subseteq E: |\Imax(A)| \geq \frac{1}{2} \cdot \E_{A \sim \calD}\left[\rank(A\cap (H\cup L))\right] \right \}.
$$

\noindent We can lower bound the probability that the event $W$ happens in the following way

\setcounter{equation}{0}
\begin{spreadlines}{2ex}
\begin{align}
    \Pr_{A \sim \calD} \left [ A \in W \right ] 
    &= \Pr_{A \sim \calD} \left [ \rank(A \cap (H\cup L)) \geq \frac{1}{2} \E_{A \sim \calD}\left[\rank(A\cap (H\cup L))\right] \right ] \\
    &\geq 1 - \Pr_{A \sim \calD} \left [ \rank(A \cap (H\cup L)) \leq \frac{1}{2} \E_{A \sim \calD}\left[\rank(A\cap (H\cup L))\right] \right ]\\
    &\geq 1- e^{-\E_{A \sim \calD}\left[\rank(A\cap (H\cup L))\right]/8}\\
    &\geq 1 - e^{-\frac{\OPT}{16}}\\
    &\geq \frac{1}{2} \tag{$\bigtriangleup$} \label{eq:matroids-case2a-constant-prob-W}
\end{align}
\end{spreadlines}

\noindent where to get (3) we applied Corollary~\ref{cor: concentration of submodular functions},
for (4) we used the property of our prefix (Eq. (\ref{eq:prefix-rank-opt})) and to get (\ref{eq:matroids-case2a-constant-prob-W}) we used that $\OPT \geq 20$.

\vspace{1em}
\noindent Our next step is to prove that for any $A \in W$, with constant probability,
one of the sampled  multisets $\tV_1, \dots, \tV_k$, has a large enough intersection with
$\Imax(A)$.

\begin{lemma}\label{lem:matroids-case2a-multisets-are-opt-probability}
    For any $A \in W$, it holds that
    $$\Pr_{A\sim \calD, \ALG} \left [ \max_{i \in [k]} |\tV_i \cap \Imax(A)| \geq \frac{1-\nicefrac{1}{e}}{1600}\cdot \OPT \right ] \geq 1-\frac{1}{e}.$$
\end{lemma}

\vspace{1em}
\begin{proof}[Proof of Lemma~\ref{lem:matroids-case2a-multisets-are-opt-probability}]
    For any $A \in W$, the probability that one sample falls into $\Imax(A)$ is
    $$
    \frac{|\Imax(A)|}{dr} \geq \frac{\E_{A \sim \calD, \ALG}[\rank(A \cap (H \cup L))]}{2 d r} \geq
    \frac{\mu (1-\nicefrac{1}{e})}{20r} \geq \frac{\mu (1-\nicefrac{1}{e})}{40r_2},
    $$
    \noindent where for the first inequality we used that $A \in W$, for the second
    we used Lemma~\ref{lem:matroids-case2-unified-lower-bound-exp-rank-in-prefix} 
    and for the third we used that
    $r_2 \geq r/2$ due to the definition of \ref{case2:definition}.

\vspace{1em}
    Therefore, for any $i \in [k]$, the random variable $|\tV_i \cap \Imax(A)|$ stochastically
    dominates a Binomial random variable that follows $\Bin(r_2,\mu(1-\nicefrac{1}{e})/40r_2)$. Furthermore, 
    for any fixed $A \in W$, the random variables $|\tV_1 \cap \Imax(A)|, \dots, |\tV_k \cap 
    \Imax(A)|$, depend only on the internal randomness of the algorithm and are, thus, independent.

    \vspace{0.7em}
    \noindent We first show that the expectation of the maximum of $k$ independent trials
    of the afforementioned Binomials, i.e. $\Bin(r_2, \mu(1-\nicefrac{1}{e})/40r_2)$, 
    is $\Theta(1) \cdot \OPT$. To do so, we define the following i.i.d. random variables:
    
    \begin{spreadlines}{2ex}
    \begin{align*}
        Z_1, \dots, Z_k &\sim \Bin \left(r_2, \frac{\mu(1-\nicefrac{1}{e})}{40 \cdot r_2}\right)\\
        Z_1', \dots, Z_k' &\sim \Bin \left (r_2, \frac{\mu}{r_2} \right)\\
        Y_1, \dots, Y_k &\sim \PB(\{ p_{b_1}, \dots, p_{b_{r_2}}\}).
    \end{align*}
    \end{spreadlines}

    \noindent Then, it is true that
    $$
    \E\left [ \max_{i \in [k]} Z_i \right] \geq \frac{(1-\nicefrac{1}{e})}{40} \E\left [ \max_{i \in [k]} Z_i' \right] \geq \frac{(1-\nicefrac{1}{e})}{40} \E\left [ \max_{i \in [k]} Y_i \right] \geq 
    \frac{(1-\nicefrac{1}{e})}{320} \cdot \OPT,
    $$
    \noindent where to get the first inequality, intuitively, one can sample the random variables $Z_i$ by first sampling
    the random variables $Z_i'$ and then discarding each trial that succeeded, independently, with probability $(1-\nicefrac{1}{e})/40$.
    In that way, for every outcome of the variables $Z_1', \dots, Z_k'$ the corresponding variables $Z_1, \dots, Z_k$
    will retrieve at least a $(1-\nicefrac{1}{e})/40$ fraction of the initial value on expectation. We formally prove this in Lemma~\ref{lem:scaled_prob_bin}. The second inequality is true since replacing every $p_{b_i}$ with the average activation
    probability maximizes the variance of the Poisson Binomial, while keeping the expectation the same, which leads to an
    increased expected maximum (Lemma~\ref{lem: the maximum entropy poisson binomila distributionis the poisson one}).
    Finally, for the third inequality we used the definition of \ref{case2:definition}.

    \vspace{1em}
    \noindent Finally, the proof of the lemma is concluded as follows
    \setcounter{equation}{0}
    \begin{align}
        \Pr_{A \sim \calD, \ALG}\left[ \max_{i \in [k]} |\tV_i \cap \Imax(A)| \geq \frac{(1-\nicefrac{1}{e})}{1600} \OPT \right] 
        &\geq \Pr\left [ \max_{i \in [k]} Z_i \geq \frac{(1-\nicefrac{1}{e})}{1600} \OPT \right]\\
        & \geq \Pr\left [ \max_{i \in [k]} Z_i \geq \frac{1}{5} \E\left [ \max_{i \in [k]} Z_i \right]  \right]\\
        & \geq 1-\frac{1}{e}
    \end{align}
    \noindent where for (1) we used the fact that each random variable $|\tV_i \cap \Imax(A)|$
    stochastically dominates each binomial $Z_j$, for (2) we used that $\E\left [ \max_{i \in [k]} Z_i \right] \geq (1-\nicefrac{1}{e})\OPT/320$ and for (3) we used a concentration inequality for the maximum of independent Binomials (Lemma~\ref{lem:max_binom_concentration}) and used the fact that $\OPT \geq 4100$ to get that $\E\left [ \max_{i \in [k]} Z_i \right]$ is at least a constant.
\end{proof}

\vspace{2em}
\noindent We finish the analysis of \ref{case2:definition} by proving the following lemma.

\begin{lemma}\label{lem:matroids-case2a-vi-is-theta-opt}    
$\E_{A \sim \calD, \ALG}\left[ \max_{i \in [k]} |V_i \cap \Imax(A)| \right] \geq (1-\nicefrac{1}{e})^2 \OPT / 3200.$
\end{lemma}

\begin{proof}[Proof of Lemma~\ref{lem:matroids-case2a-vi-is-theta-opt}]
    For any $A' \in W$, by the law of total expectation we have that
    \begin{align*}
    \E_{A \sim \calD, \ALG}&\left[ \max_{i \in [k]} |V_i \cap \Imax(A)| \Big| A = A'\right]=\\
    &= \sum_{x = 0}^r \E_{A \sim \calD, \ALG}\left[ \max_{i \in [k]} |V_i \cap \Imax(A)| \Big | \max_{i \in [k]} |\tV_i \cap \Imax(A)| = x, A = A' \right] \cdot\\
    &\hspace{6em} \cdot \Pr_{A \sim \calD, \ALG} \left [ \max_{i \in [k]} |\tV_i \cap \Imax(A)| = x \Big | A = A'\right]
    \end{align*}

    \vspace{1em}
    \noindent Therefore, we are interested in analysing the following expression
    $$
    \E_{A \sim \calD, \ALG}\left[ \max_{i \in [k]} |V_i \cap \Imax(A)| \Big | \max_{i \in [k]} |\tV_i \cap \Imax(A)| = x, A = A'\right].
    $$

    The above can be seen as a balls and bins question. To be more precise, we know that one of
    the sampled multisets of the algorithm has sampled $x$ elements from $\Imax(A')$ and we
    are interested in calculating the expected value of the number of unique elements that were
    sampled. This is the same as throwing $x$ balls uniformly at random into $|\Imax(A')|$ bins
    and calculating the expected number of non-empty bins. Therefore, from Lemma~\ref{lem: lower bound on expected number of full bins },
    for all $x \in \{0,\dots,r\}$ and all $A' \in W$, we get that

    \begin{align}
        \E_{A \sim \calD, \ALG}\left[ \max_{i \in [k]} |V_i \cap \Imax(A)| \Big | \max_{i \in [k]} |\tV_i \cap \Imax(A)| = x, A = A'\right] \geq \min \left\{ \frac{x}{2}, \frac{3 |\Imax(A')|}{10}\right\} \tag{$\nabla$} \label{eq:matroids-case2-balls-bins}
    \end{align}

    \noindent Using the latter, for any $A' \in W$, we get that

    \setcounter{equation}{0}
    \begin{spreadlines}{2ex}
    \begin{align}
       & \E_{\ALG}\left[ \max_{i \in [k]} |V_i \cap \Imax(A)| \Big | A = A'\right] = \\
       &\geq \sum_{x} \min \left\{ \frac{x}{2}, \frac{3 |\Imax(A')|}{10}\right\} \cdot 
       \Pr_{\ALG}\left[ \max_{i \in [k]} |\tV_i \cap \Imax(A')| = x \right]\\
       & \geq \min \left\{ \frac{1-\nicefrac{1}{e}}{1600} \OPT, \frac{3 |\Imax(A')|}{10} \right\}
       \cdot \Pr_{\ALG} \left [ \max_{i \in [k]} |\tV_i \cap \Imax(A')| \geq \frac{1-\nicefrac{1}{e}}{1600}\cdot \OPT \right ]\\
       & \geq \frac{(1-\nicefrac{1}{e})^2}{1600} \cdot \OPT,
    \end{align}
    \end{spreadlines}

    \noindent where for (2) we used Eq.(\ref{eq:matroids-case2-balls-bins}), to
    get (3) we restricted the summation to $x \geq (1-\nicefrac{1}{e})\OPT/1600$
    and to get (4) we used Lemma~\ref{lem:matroids-case2a-multisets-are-opt-probability}
    and the fact that $|\Imax(A')| \geq \OPT/4$ for any $A' \in W$ due
    to the definition of W and the property of our prefix (Eq. (\ref{eq:prefix-rank-opt})).

\vspace{1em}
    \noindent Finally, since the above holds for any $A' \in W$, we conclude the proof of the lemma
    by applying the law of total expectation:
    \setcounter{equation}{0}
    \begin{spreadlines}{2ex}
    \begin{align}
        &\E_{A \sim \calD, \ALG}\left[ \max_{i \in [k]} |V_i \cap \Imax(A)| \right] \geq\\
        &\geq
        \sum_{A' \in W } \E_{A \sim \calD, \ALG}\left[ \max_{i \in [k]} |V_i \cap \Imax(A)| \Big | A = A' \right] \cdot \Pr_{A \sim \calD}\left [ A = A' \right]\\
        &\geq \frac{(1-\nicefrac{1}{e})^2}{3200} \cdot \OPT \cdot \sum_{A' \in W} \Pr_{A \sim \calD}\left [ A = A' \right]\\
        & \geq \frac{(1-\nicefrac{1}{e})^2}{1600} \cdot \OPT \cdot \Pr_{A \sim \calD}\left [ A \in W \right]\\
        & \geq \frac{(1-\nicefrac{1}{e})^2}{3200} \cdot \OPT,
    \end{align}
    \end{spreadlines}

    \noindent where for (3) we used the lower bound proved above, for (4) we used a union-bound and for (5) we used Eq. (\ref{eq:matroids-case2a-constant-prob-W}).
\end{proof}

\newpage
\bibliographystyle{alpha} 
\bibliography{ref}

\newcommand{\etalchar}[1]{$^{#1}$}
\begin{thebibliography}{DMVW23}

\bibitem[Bal20]{data-driven-survey}
Maria-Florina Balcan.
\newblock Data-driven algorithm design, 2020.

\bibitem[BMKS16]{balkanski-data-summarization}
Eric Balkanski, Baharan Mirzasoleiman, Andreas Krause, and Yaron Singer.
\newblock Learning sparse combinatorial representations via two-stage submodular maximization.
\newblock In {\em International Conference on Machine Learning}, pages 2207--2216. PMLR, 2016.

\bibitem[Bru69]{brualdi1969comments}
Richard~A Brualdi.
\newblock Comments on bases in dependence structures.
\newblock {\em Bulletin of the Australian Mathematical Society}, 1(2):161--167, 1969.

\bibitem[CSVZ22]{graph-learned-primal-dual}
Justin Chen, Sandeep Silwal, Ali Vakilian, and Fred Zhang.
\newblock Faster fundamental graph algorithms via learned predictions.
\newblock In {\em International Conference on Machine Learning}, pages 3583--3602. PMLR, 2022.

\bibitem[CVZ11]{crspaper2011}
Chandra Chekuri, Jan Vondr\'{a}k, and Rico Zenklusen.
\newblock Submodular function maximization via the multilinear relaxation and contention resolution schemes.
\newblock In {\em Proceedings of the Forty-Third Annual ACM Symposium on Theory of Computing}, STOC '11, page 783–792, New York, NY, USA, 2011. Association for Computing Machinery.

\bibitem[DBC{\etalchar{+}}24]{march-madness}
Jeff Decary, David Bergman, Carlos Cardonha, Jason Imbrogno, and Andrea Lodi.
\newblock The madness of multiple entries in march madness, 2024.

\bibitem[DIL{\etalchar{+}}21]{faster-matchings}
Michael Dinitz, Sungjin Im, Thomas Lavastida, Benjamin Moseley, and Sergei Vassilvitskii.
\newblock Faster matchings via learned duals.
\newblock {\em Advances in neural information processing systems}, 34:10393--10406, 2021.

\bibitem[DMVW23]{ford-fulkerson}
Sami Davies, Benjamin Moseley, Sergei Vassilvitskii, and Yuyan Wang.
\newblock Predictive flows for faster ford-fulkerson.
\newblock In {\em International Conference on Machine Learning}, pages 7231--7248. PMLR, 2023.

\bibitem[Dug20]{dughmi2020}
Shaddin Dughmi.
\newblock The outer limits of contention resolution on matroids and connections to the secretary problem.
\newblock In {\em 47th International Colloquium on Automata, Languages, and Programming (ICALP 2020)}. Schloss-Dagstuhl-Leibniz Zentrum f{\"u}r Informatik, 2020.

\bibitem[Dug22]{dughmi2022}
Shaddin Dughmi.
\newblock Matroid secretary is equivalent to contention resolution.
\newblock {\em Innovations in Theoretical Computer Science (ITCS)}, 2022.

\bibitem[DVW24]{push-relabel}
Sami Davies, Sergei Vassilvitskii, and Yuyan Wang.
\newblock Warm-starting push-relabel.
\newblock {\em arXiv preprint arXiv:2405.18568}, 2024.

\bibitem[Fei98]{feige1998threshold}
Uriel Feige.
\newblock A threshold of ln n for approximating set cover.
\newblock {\em Journal of the ACM (JACM)}, 45(4):634--652, 1998.

\bibitem[GMS23]{robust-portfolios}
Swati Gupta, Jai Moondra, and Mohit Singh.
\newblock Balancing notions of equity: Approximation algorithms for fair portfolio of solutions in combinatorial optimization.
\newblock {\em arXiv preprint arXiv:2311.03230}, 2023.

\bibitem[KPR04]{kleinberg2004segmentation}
Jon Kleinberg, Christos Papadimitriou, and Prabhakar Raghavan.
\newblock Segmentation problems.
\newblock {\em Journal of the ACM (JACM)}, 51(2):263--280, 2004.

\bibitem[KR18]{kleinberg2018team}
Jon Kleinberg and Maithra Raghu.
\newblock Team performance with test scores.
\newblock {\em ACM Transactions on Economics and Computation (TEAC)}, 6(3-4):1--26, 2018.

\bibitem[LSV23]{bellman-ford}
Silvio Lattanzi, Ola Svensson, and Sergei Vassilvitskii.
\newblock Speeding up bellman ford via minimum violation permutations.
\newblock In {\em International Conference on Machine Learning}, pages 18584--18598. PMLR, 2023.

\bibitem[MV22]{algorithms-with-predictions-survey}
Michael Mitzenmacher and Sergei Vassilvitskii.
\newblock Algorithms with predictions.
\newblock {\em Communications of the ACM}, 65(7):33--35, 2022.

\bibitem[NWF78]{nemhauser1978analysis}
George~L Nemhauser, Laurence~A Wolsey, and Marshall~L Fisher.
\newblock An analysis of approximations for maximizing submodular set functions—i.
\newblock {\em Mathematical programming}, 14:265--294, 1978.

\bibitem[QS22]{singla_submodular_dominance}
Frederick Qiu and Sahil Singla.
\newblock Submodular dominance and applications.
\newblock In {\em Approximation, Randomization, and Combinatorial Optimization. Algorithms and Techniques (APPROX/RANDOM 2022)}. Schloss Dagstuhl-Leibniz-Zentrum f{\"u}r Informatik, 2022.

\bibitem[Sin18]{singla_thesis}
Sahil Singla.
\newblock Combinatorial optimization under uncertainty: Probing and stopping-time algorithms.
\newblock {\em Unpublished doctoral dissertation, Carnegie Mellon University}, 2018.

\bibitem[SZKK17]{krause-data-summarization}
Serban Stan, Morteza Zadimoghaddam, Andreas Krause, and Amin Karbasi.
\newblock Probabilistic submodular maximization in sub-linear time.
\newblock In {\em International Conference on Machine Learning}, pages 3241--3250. PMLR, 2017.

\bibitem[VCZ11]{vondrak2011submodular}
Jan Vondr{\'a}k, Chandra Chekuri, and Rico Zenklusen.
\newblock Submodular function maximization via the multilinear relaxation and contention resolution schemes.
\newblock In {\em Proceedings of the forty-third annual ACM symposium on Theory of computing}, pages 783--792, 2011.

\bibitem[Von10]{concentrationofsubmodularfunctions}
Jan Vondrak.
\newblock A note on concentration of submodular functions, 2010.

\end{thebibliography}

\newpage
\appendix
\section{Useful Lemmas}

\subsection{Assumption 3}
\begin{lemma}\label{lem: assumption 3}
    Let $\calJ = (M, k, \calD)$ be a triplet describing an instance of our problem and  $\calB = \{ B_1, \dots, B_k \}$ a portfolio containing $k$ disjoint bases of maximum weight. If $\val (\calO) < 4100$ then $\val (\calB) \geq \Theta (1) \cdot \val (\calO)$.
\end{lemma}
\begin{proof}
Let $\mu = \sum_{e \in \bigcup_{i = 1}^k B_i} p_e $ and note that $\val (\calO ) \leq \sum_{e \in \bigcup_{i = 1}^k O_i} p_e \leq \mu$ where the second inequality comes from the definition of bases $B_1, \dots, B_k$.
Thus, $ \val (\calO ) \leq \min \left\{  \mu, 4100 \right\} $.
We consider the cases $\mu \geq 1/2 $ and $\mu < 1/2$ separately.

If $\mu \geq 1/2$ then:
\begin{align*}
    \val (\calB) &\geq \Pr \left[ A \cap \bigcup_{i = 1}^k B_i \neq \emptyset\right]\\
    &= 1-  \Pr \left[ A \cap \bigcup_{i = 1}^k B_i = \emptyset\right]\\
     &= 1-  \prod_{e \in \bigcup_{i = 1}^k B_i } \Pr \left[e \not \in  A \right]\\
      &= 1-  \prod_{e \in \bigcup_{i = 1}^k B_i } \left(1 - p_e\right)\\
        &\geq 1-  \prod_{e \in \bigcup_{i = 1}^k B_i } e^{- p_e}\\
        &= 1-   e^{- \mu}\\
        &\geq 1-   e^{- 1/2}\\
        &\geq 0.3\\
\end{align*}
We now turn our attention to the case where $\mu < 1/2$:
\begin{align*}
    \val (\calB) &\geq \Pr \left[ \abs{A \cap \bigcup_{i = 1}^k B_i} = 1\right]\\
    &= \sum_{e \in \bigcup_{i = 1}^k B_i } \Pr \left[e \in  A \right] \cdot  \prod_{\bigcup_{i = 1}^k B_i \ni e' \neq e  } \Pr \left[e' \not \in  A \right]\\
  &= \sum_{e \in \bigcup_{i = 1}^k B_i } p_e \cdot  \prod_{\bigcup_{i = 1}^k B_i \ni e' \neq e  }(1 - p_{e'})\\
   &\geq \sum_{e \in \bigcup_{i = 1}^k B_i } p_e \cdot  \prod_{\bigcup_{i = 1}^k B_i \ni e' \neq e  } e^{-2 \cdot p_{e'}}\\
    &= \sum_{e \in \bigcup_{i = 1}^k B_i } p_e \cdot  e^{-2 \cdot ( \mu - p_{e}) }\\
     &= e^{-2 \cdot \mu} \sum_{e \in \bigcup_{i = 1}^k B_i } p_e \cdot  e^{2 \cdot  p_{e} }\\
      &\geq e^{-1} \sum_{e \in \bigcup_{i = 1}^k B_i } p_e \\
       &= e^{-1} \cdot \mu \\
\end{align*}
where in the second inequality we used that $(1 - x) \geq e^{-2x}$ for $x \in [0,1/2]$
\end{proof}

\subsection{Matroids}\label{subsec:appendix-matroids}

\begin{corollary}[Corollary 3 in~\cite{brualdi1969comments}]\label{thm:bij} Let $M=(E, \calI)$ be a matroid, and let $B$ and $B'$ be bases of $M$. Then there exists a bijection $\pi: B' \rightarrow B$ such that $$(B \setminus \{\pi(e)\}) \cup \{e\}$$ is a basis for all $e \in B$. Furthermore, such a bijection can be found in polynomial time with respect to the size of the matroid. 
\end{corollary}

\vspace{1em}
\begin{corollary}[Theorem 4.8 \cite{vondrak2011submodular}]\label{cor: Expected rank(A) is at least (1-1/e) Expected size of A}
Let $M = (E, \calI)$ be a matroid, $\{ X_e\}_{e \in E}$ be a collection of independent random variables that take value in $\{0,1\}$, ${\bf p}  \in {[0, 1]}^{\abs{E}}$ be the corresponding probability vector, i.e., $p_e = \Pr [X_e = 1]$ and $S = \{e \in E: X_e = 1 \}$ is the subset of $E$ containing all variables such that their corresponding random variable is $1$. If $\bf p$ is in the matroid polytope, i.e.,  ${\bf p} \in P(M) $, then:
\begin{align*}
    \E [\rank(S)]  \geq  (1 - 1/e)  \E [\abs{S}]]
\end{align*}
\end{corollary}

\subsection{Probability lemmas}

\vspace{1em}
\begin{corollary}[Chernoff Bounds]\label{cor:chernoff}
Let $X= \sum_{i=1}^n X_i$ be the sum of $n$ independent Bernoulli random variables, where $\Pr[X_i=1]=p_i$, and $\mu = \sum_{i=1}^{n} p_i$. Then,

\begin{align*}
  &\Pr[X \geq (1+\delta) \mu] \leq e^{-\mu \delta^2/(\delta+3)} &&\quad \quad \forall \delta \geq 0\\
  &\Pr[X \leq (1-\delta) \mu] \leq e^{-\delta^2 \mu/2}
  &&\quad \quad \forall \delta\in [0,1]\\
  &\Pr[|X-\mu| \geq \delta\mu] \leq 2e^{-\delta^2 \mu/3}
  &&\quad \quad \forall \delta\in [0,1]
\end{align*}
\end{corollary}

\vspace{1em}
\begin{corollary}[Application of Corollary 3.2. in~\cite{concentrationofsubmodularfunctions} to the $\rank$ function]\label{cor: concentration of submodular functions}
Let $M = (E, \calI)$ be a matroid, $\{ X_e\}_{e \in E}$ be a collection of independent random variables that take values in $\{0,1\}$ and $S = \{e \in E: X_e = 1 \}$ be the subset of $E$ containing all variables such that their corresponding random variable is equal to $1$. Then for any $\delta > 0$:
\begin{align*}
&\Pr[\rank(S) \leq (1-\delta) \E [\rank(S)]] \leq e^{-\delta^2 \E [\rank(S)]/2}\\
  &\Pr[\rank(S) \geq (1+\delta) \E [\rank(S)]] \leq \left(\frac{e^\delta}{(1 + \delta)^{1 + \delta}}\right)^{ \E [\rank(S)]}
\end{align*}
\end{corollary}

    

\vspace{1em}
\begin{lemma}\label{lem:binomial-cdf-decay}
    Let $X \sim \Bin(n,p)$ and $t \in \nats$. The following inequality holds
    $$
    \Pr[X \geq 2t] \leq \left ( \Pr[X \geq t]\right)^2.
    $$
\end{lemma}

\begin{proof}[Proof of Lemma~\ref{lem:binomial-cdf-decay}]
Let $X = \sum_{i=1}^n X_i$, where $X_i \sim \Be(p)$.
We will use $X_{i:j}$ to denote the partial sum
$\sum_{k=i}^j X_{k}$. 

\vspace{1em}
\noindent We define a new random variable $U$ that takes values in $\{t,\dots,n+1\}$ as follows: $U$ is the minimum index $j$ such that $\sum_{i=1}^j X_i= t$ when $X \geq t$
and takes the value $n+1$ when $X < t$. 
Notice that the events space of $X$ can now be written as the following partition: $\cup_{i=t}^{n+1} (U = i)$. Using $U$ we can 
rewrite the left-hand side of the desired inequality as

\begin{spreadlines}{2ex}
\begin{align*}
    \Pr[X \geq 2t] & = \sum_{i = t}^n \Pr[X_{i+1:n} \geq t \wedge U = i]& \\
    & = \sum_{i = t}^n \Pr[X_{i+1:n} \geq t] \Pr[U = i]  & \hfill \text{$(U=i)$ is independent from $X_{i+1:n}$}\\
    & \leq \Pr[X \geq t] \sum_{i = t}^n \Pr[U = i] & \hfill \text{$X$ stochastically dominates $X_{i+1:n}$}\\
    & = \Pr[X \geq t] \Pr\left[\cup_{i=t}^{n} (U = i)\right] & \hfill \text{$(U=i)$ are mutually exclusive}\\[2ex]
    & = \Pr[X \geq t] \Pr[X \geq t] & \hfill \text{by definition of $U$}\\
    & = \left(\Pr[X \geq t]\right)^2 &
\end{align*}
\end{spreadlines}

\end{proof}

\vspace{1em}
\vspace{2em}
\begin{lemma}\label{lem:max_binom_concentration} 
Let $B_1, \dots, B_k$ be i.i.d. random variables with $B_i \sim \Bin(n, p)$ and $\hat{B} = \max_{i=1}^{k}B_i$. Assuming that $k \geq 4$ and $\E[\hat{B}] \geq 30$, it holds that
$\Pr \left[\hat{B} \geq \E[\hat{B}]/6 \right] \geq 1-1/e$.
\end{lemma}

\begin{proof}[Proof of Lemma~\ref{lem:max_binom_concentration}]
Let $t$ be the largest natural number such that $\Pr[B_1 \geq t] \leq 1/k$. 
The lemma follows immediately from the following statements and the assumption that
$\E[\hat{B}] \geq 30$.

\setcounter{equation}{0}
\begin{align}
\Pr[\hat{B} \geq t-1] \geq 1-\frac{1}{e}  \label{eq:prob-at-least-t}\\
\E[\hat{B}] \leq 5t  \label{eq:ub-for-expectation}
\end{align}

\noindent We will start by proving Equation~\ref{eq:prob-at-least-t}. By the definition of $t$, we know that $\Pr[B_1 \geq t - 1] \geq 1/k$. Therefore,

\begin{align*}
    \Pr[\hat{B} \geq t-1] & = 1 - \Pr[\hat{B} < t-1]\\
    & = 1 - \Pr[\cap_{i=1}^k (B_i < t-1)]\\
    & = 1 - (\Pr[B_1 < t-1])^k\\
    & = 1 - (1-\Pr[B_1 \geq t-1])^k\\
    & \geq 1-\left(1-\frac{1}{k}\right)^k\\
    & \geq 1-\frac{1}{e}
\end{align*}

\noindent Finally, to prove Equation~\ref{eq:ub-for-expectation}
we can upper bound the expectation of $\hat{B}$ as follows
\begin{align*}
    \E[\hat{B}] &= \sum_{i = 0}^n \Pr[\hat{B} \geq i]\\
    &=\sum_{i = 0}^{t-1} \Pr[\hat{B} \geq i] + \sum_{i = t}^n \Pr[\hat{B} \geq i] \\
    &\leq t + \sum_{i = t}^n \Pr[\hat{B} \geq i]\\
    & \leq t + \sum_{j = 0}^{\infty} \sum_{i = 2^j t}^{2^{j+1}t} \Pr[\hat{B} \geq i]\\
    & \leq t + t \sum_{j = 0}^{\infty} 2^j \Pr[\hat{B} \geq 2^j t]\\
    & \leq t + t \sum_{j = 0}^{\infty} \left(\frac{2}{e}\right)^j\\
    & = t + t \left(\frac{e}{e-2}\right)\\
    & \leq 5t
\end{align*}
\noindent where we used that $\Pr[\hat{B} \geq 2^j t] \leq 1/e^j$ due to Lemma~\ref{lem:binomial-cdf-decay} and the definition of $t$.
\end{proof}

\vspace{1em}
\begin{lemma}\label{lem:poisson-binomial-random-subset}
Let $Y_1, \dots, Y_k$ be i.i.d. Poisson Binomial random variables following
$\PB(\{p_1, \dots, p_n\})$ and $X_{11}, \dots, X_{nk}$ be independent Bernoulli random variables
with $X_{ij} \sim Be(p_i) $ for all $i \in [n], j \in [k]$. Let also $S_1, \dots, S_n$ be independent
Bernoulli random variables with  $S_i  \sim Be(q_i)$ where $q_i \geq c, \forall  i \in [n]$. Then:
\begin{align*}
    \E \left[\max_{j \in [k]} \sum_{i=1}^n S_i X_{ij} \right] \geq c \cdot \E \left[ \max_{j \in [k]} Y_j \right]
\end{align*}
\end{lemma}

\vspace{1em}
\begin{proof}[Proof of Lemma~\ref{lem:poisson-binomial-random-subset}]
Note that $\forall j \in [k]$ we have that: $\sum_{i=1}^n  X_{ij}  \sim \PB(\{p_1, \dots, p_n\})$. Thus, it is enough to prove that:
\begin{align*}
    \E \left[\max_{j \in [k]} \sum_{i=1}^n S_i X_{ij} \right] \geq c \cdot \E \left[ \max_{j \in [k]} \sum_{i=1}^n X_{ij} \right]
\end{align*}
Let $j^* = \argmax_{j \in [k]} \sum_{i=1}^n X_{ij}$, $ j_S^* = \argmax_{j \in [k]} \sum_{i=1}^n S_i X_{ij} $ be the random variables which denote the index of the highest sum in the right and left hand side respectively.
We then have:
\setcounter{equation}{0}
\begin{align}
    \E \left[\max_{j \in [k]} \sum_{i=1}^n X_{ij} \right] &= \sum_{j \in [k]}  \Pr \left[ j^* = j\right] \cdot  \E \left[\sum_{i=1}^n X_{ij} \mid j^* = j\right]\\
    & =  \sum_{j \in [k]} \sum_{i=1}^n  \Pr \left[ j^* = j\right] \cdot  \E \left[ X_{ij} \mid j^* = j\right]\\
    & \leq  \sum_{j \in [k]} \sum_{i=1}^n  \Pr \left[ j^* = j\right] \cdot \frac{\E [S_i]}{c} \E \left[ X_{ij} \mid j^* = j\right]\\
    & =  \frac{1}{c}  \sum_{j \in [k]} \sum_{i=1}^n  \Pr \left[ j^* = j\right] \cdot  \E \left[S_i X_{ij} \mid j^* = j\right]\\
    & = \frac{1}{c} \E \left[ \sum_{i=1}^n S_i X_{ij^*} \right]\\
    & \leq \frac{1}{c} \E \left[ \sum_{i=1}^n S_i X_{ij_S^*} \right]\\
    & = \frac{1}{c} \E \left[ \max_{j \in [k]} \sum_{i=1}^n S_i X_{ij} \right]
\end{align}
Where in (1) we used the law of total expectation,
from (1) to (2) the linearitly of expectation,
from (2) to (3) that $\E [S_i] = q_i \geq c, \forall i \in [n]$,
from (3) to (4) the fact that random variables $S_i$ are independent from random variables $X_{ij}$ and consequently also independent from the event $\{j^* = j \}$,
from (4) to (5) we again use both the law of total expectation and the linearity of expectation, and
from (5) to (6) we used the definition of random variable $ j_S^*$.
\end{proof}

\vspace{1em}

\begin{lemma}\label{lem: the maximum entropy poisson binomila distributionis the poisson one}
Let $\calY$ be a Poisson binomial distributions with parameters $\bf y \in [0,1]^n$ and $\calY^{avg}$ be a Poisson distribution with parameters $\nicefrac{\sum_{i=1}^n y_i}{n}$ and $n$. 
Let $Y_1, \dots, Y_k \sim \calY$ and $Y^{avg}_1, \dots, Y^{avg}_k \sim \calY^{avg}$ be two sets of $k$ independent variables each. Then
\begin{align*}
    \E [\max_{i \in [k]} Y_i] \leq  \E [\max_{i \in [k]} Y_i^{avg}] 
\end{align*}
\end{lemma}

\begin{proof}
Note that if $y_1 = y_2 =\dots = y_n$ then the lemma trivially holds. 


\noindent Without loss of generality we assume that $y_1 < y_2$. Let $\calY'$ be a Poisson binomial distribution with parameters $\{ \frac{y_1 + y_2}{2}, \frac{y_1 + y_2}{2}, y_3, \dots, y_n \}$ and $Y_1', \dots, Y_k' \sim \calY'$ be $k$ independent variables from that distribution. We argue that:
\begin{align*}
    \E [\max_{i \in [k]} Y_i] \leq  \E [\max_{i \in [k]} Y_i'] 
\end{align*}
Note that the latter suffices to prove the lemma as we can repeat the same arguments as long as not all parameters of the new Poisson binomial are equal.

\vspace{0.5em}
Let $X_{ij} \sim Be(y_i)$ for all $i \in [n], j \in [k]$ 
and  $C_j \sim Be \left( \frac{y_1 +y_2}{2}\right)$ and $L_j \sim Be \left( \frac{y_1 + y_2}{2}\right)$ for all $j \in [k]$ be independent random variables. Then we define:
\begin{align*}
&Y_j = \sum_{i = 1}^n X_{ij}\\
&Y_j ' = L_j + C_j + \sum_{i = 3}^n X_{ij}    
\end{align*}
Then
\begin{align*}
    &\E \left[\max_{j \in [k]} Y_j \right] \leq  \E \left[\max_{j \in [k]} Y_j ' \right] \Longleftrightarrow\\
    &\E \left[\max_{j \in [k]} \sum_{i =1}^n X_{ij} \right] \leq  \E \left[\max_{j \in [k]} \left( L_j + C_j + \sum_{i =3}^n X_{ij}\right) \right] \Longleftrightarrow\\
     &\E \left[\max_{j \in [k]}  \left( X_{1j} + X_{2j} \right) \right] \leq  \E \left[\max_{j \in [k]} \left(  L_j + C_j \right) \right] \overset{y_{12} = \frac{y_1 + y_2}{2}}{\Longleftrightarrow}\\
      &2 - (1 - y_1 y_2)^k - (1 - y_1)^k (1 - y_2)^k \leq  2 - (1 - y_{12}^2)^k - (1 - y_{12})^{2k} \Longleftrightarrow\\
     &(1 - y_1 y_2)^k + (1 - y_1)^k (1 - y_2)^k \geq  (1 - y_{12}^2)^k + (1 - y_{12})^{2k} \Longleftrightarrow\\
    &(1 - y_1 y_2)^k + (1 - y_1)^k (1 - y_2)^k \geq  (1 - y_{12}^2)^k + (1 - y_{12})^{2k} \overset{t = \frac{y_2 - y_1}{2}}{\Longleftrightarrow}\\
    &(1 - y_{12}^2 + t^2)^k + (1 - y_{12} + t)^k (1 - y_{12} - t)^k \geq  (1 - y_{12}^2)^k + (1 - y_{12})^{2k} \Longleftrightarrow\\
    &(1 - y_{12}^2 + t^2)^k + ((1 - y_{12})^2 - t^2)^k  \geq  (1 - y_{12}^2)^k + (1 - y_{12})^{2k}\tag{*}\label{eq: some ineq}
\end{align*}

Where we explicitly calculate the maximum of the $k$ random variables $X_{1j} + X_{2j}$ using that:
\begin{align*}
    \E \left[\max_{j \in [k]}  \left( X_{1j} + X_{2j} \right) \right] &= 
    \Pr \left[\max_{j \in [k]}  \left( X_{1j} + X_{2j} \right) \geq 1\right] +
    \Pr \left[\max_{j \in [k]}  \left( X_{1j} + X_{2j} \right) \geq 2\right]\\
     &= 2 -
    \Pr \left[\max_{j \in [k]}  \left( X_{1j} + X_{2j} \right) < 1\right] 
    -\Pr \left[\max_{j \in [k]}  \left( X_{1j} + X_{2j} \right) < 2\right]\\
    &= 2 -
    \Pr \left[\max_{j \in [k]}  \left( X_{1j} + X_{2j} \right) = 0\right] 
    -\Pr \left[\max_{j \in [k]}  \left( X_{1j} + X_{2j} \right) < 2\right]\\
     &=2-
     \prod_{j \in [k]} \Pr \left[X_{1j} + X_{2j}  = 0\right] 
    -\prod_{j \in [k]} \Pr \left[ X_{1j} + X_{2j} < 2\right]\\
    &=2 -  (1 - y_1)^k (1 - y_2)^k - (1 - y_1 y_2)^k \\
\end{align*}
and similarly for the max of the $k$ random variables $L_j + C_j$.

\vspace{0.5em}
To prove~\cref{eq: some ineq}, for any constant $y \in [0,1]$ we define the function $f(\cdot)$ and prove that achieves its minimum at $x = 0$:
\begin{align*}
    f(x) = (1 - y^2 + x)^k + ((1 - y)^2 - x)^k, x \in [0,1-y]
\end{align*}

Note that since $1 - y^2 \geq (1 - y)^2 \Longleftrightarrow 1 - y^2 \geq 1 + y^2 - 2y \Longleftrightarrow y \geq y^2 $ is true for any $y \in [0,1]$ and $x \in [0,1-y]$ we have that:
$f'(x) = k (1 - y^2 + x)^{k-1} -k ((1 - y)^2 - x)^{k-1} \geq 0, \forall x \in [0,1-y]$
\end{proof}

 \begin{lemma}\label{lem:PB-domination-more-trials} Let $X_1, \dots X_k$ be i.i.d. $\PB(p_1, \dots, p_{r_1})$ random variables, and $Y_1, \dots Y_k$ be i.i.d. $ \PB(q_1, \dots, q_{r_2})$. If for all $i \in [r_1], i' \in [r_2]$ we have $p_i \geq q_{i'}$ and $$\E\left[\max_{j \in [k]}Y_j\right] \geq \E\left[\max_{j\in [k]}X_j\right],$$ then $r_1 < r_2$.
\end{lemma}
\begin{proof}[Proof of Lemma~\ref{lem:PB-domination-more-trials}]
Suppose that $r_1 \geq r_2$, and for $j \in [k]$ let
$$X_j = \sum_{i=1}^{r_1}U^j_i,$$
where for $i \in [r_2]$ $$U^j_i \sim \Be(p_i)$$ and $$Y_i = \sum_{j=1}^{r_2}V^j_i,$$ where for $i \in [r_2]$ $V^j_i \sim \Be(q_i)$.

\vspace{1em}
\noindent For each $i \in [r_1]$ and $j \in [k]$ let
$$Z^i_j = \begin{cases}
    U^j_i & i \leq r_2, U^j_i = 1\\
    A^j_i & i \leq r_2, U^j_i = 0 \\
    V^j_i & \text{ otherwise}
\end{cases},$$

\noindent where $A^1_1, \dots, A^k_{r_2}$ are i.i.d. $\Be(p_i - q_i)$ random variables, which are also independent to each of $X_1, \dots, X_k, Y_1, \dots, Y_k$.

\vspace{1em}
\noindent For $j \in [k]$ let $Y_j'=\sum_{i=1}^{r_1}Z^i_j$ and observe that $Y_1', \dots, Y_k'$ are i.i.d $\PB(p_1, \dots, p_{r_1})$ random variables.

\vspace{1em}
\noindent Fix an outcome $u$ of all variables in $\calU =\{U^j_i: i \in [r_1], j \in [k]\}$, and let $\Omega$ be the set of all outcomes of $\calU \cup \{Z^i_j: i \in [r_1], j \in [k]\}$ that agree with $u$ on variables in $\calU$.

\vspace{1em}
\noindent Then,
$$\E_{\Omega}\left[\max_{j \in [k]} Y_j'\right] \geq \E_{\Omega}\left[\max_{j \in [k]} X_j\right]$$ concluding the proof.
  
\end{proof}

\vspace{1em}

\begin{lemma}\label{lem:binomials-bigger-p-better} Let $X_1, \dots, X_k$ be $\PB(p_1, \dots, p_n)$ random variables, and let $Y_1, \dots, Y_k$ be $\PB(q_1, \dots, q_n)$ random variables, such that there exists ground sets $\calU=\{U_1, \dots, U_\ell\}$ and $\calV=\{V_1, \dots, V_\ell\}$ each of i.i.d. Bernoulli random variables satisfying the property such that for each $j=1, \dots, k$ there exists some set $\Sigma_j \subseteq[\ell]$ such that $X_j = \sum_{i \in \Sigma_j} U_i$ and $Y_j=\sum_{i \in \Sigma_j} V_i$.
Then if $p_{m} \geq q_m$ for each $m \in [n]$  $$\E\left[\max_{j \in [k]}X_j\right] \geq \E\left[\max_{j \in [k]}Y_j\right].$$
\end{lemma}
\begin{proof}
If $p_m=q_m$ for all $m \in [n]$, then $X_1, \dots, X_k$ and $Y_1, \dots, Y_k$ are identically distributed.
Suppose without loss of generality that $p_{1} > q_1$ and let $$V_{1}' = 
\begin{cases}
    V_1 & V_1 = 1\\
    Z & \text{ otherwise}.
\end{cases},$$
where $Z \sim \Be(p_1-q_1)$ is independent to all variables in $\calV$.
Note that $V_{1'} \sim \Be(p_1)$ and is independent to all variables in $\calV \setminus \{V_1\}$.
For $j \in [k]$ define $Y_j'$ to be $Y_j - V_1 + V_1'$ if $Y_j$ depends on $V_1$, otherwise define $Y_j'$ to be $Y_j$.
Fix an outcome $v$ of all the variables in $\calV$, and
let $y_j$ be the value $Y_j$ takes on outcome $v$. Suppose $a = \max_{j \in [k]}y_j$, and that
$Y_1, \dots Y_{n'}$ is the subset of $\{Y_1, \dots, Y_n\}$ satisfying $y_j=a$.
Let $\Omega$ be the set of all outcomes of $\calV \cup \{V_1'\}$ that agree with $v$ on $\calV$.

\vspace{1em}
\noindent
If for every index $j \in [n']$,  $Y_j$ does not depend on $V_1$ then $$\E_{\Omega}\left[\max_{j \in [k]}Y_j'\right] = \E_{\Omega}\left[\max_{j \in [k]}Y_j\right].$$
Therefore we assume that at least one of $Y_1, \dots, Y_m$ depend on $V_1$. In this case,$$\E_{\Omega}\left[\max_{j \in [k]}Y_j\right] = a$$ and 
$$\E_{\Omega}\left[\max_{j \in [k]}Y_j\right] = a + p_1-q_1,$$ concluding the proof.



\end{proof}

\vspace{1em}

\begin{lemma}\label{lem:binomials-independent-better} Let $X_1, \dots, X_k$ be i.i.d $\PB(p_1, \dots,p_n)$ random variables for some $n \in \mathbb{N}$, and let $Y_1, \dots, Y_k$ be a set of $\PB(p_1, \dots, p_n)$ random variables, such that there exists a ground set $\calZ$ of i.i.d Bernoulli random  variables, where for each $j \in [n]$ we have $Y_j=\sum_{Z \in \calZ_j}Z$ for some $\calZ_j \subseteq \calZ$. Then $$\E\left[\max_{j \in [k]}X_j\right] \geq \E\left[\max_{j \in [k]}Y_j\right].$$
\end{lemma}
\begin{proof}
Fix an outcome $z$ of all variables in $\calZ \setminus \{Z^*\}$ for some $Z^* \in \calZ$.
Suppose that $Z^* \sim \Be(p_1)$, and that $Z^*_1, \dots, Z^*_n$ are i.i.d. $\Be(p_1)$ random variables that are also independent to each element of $\calZ$. For $j \in [n]$, define $Y'_j$ to be $Y_j-Z^*+Z_j$ if $Y_j$ depends on $Z^*$, otherwise define $Y_j'$ to be equal to $Y_j$.
In addition, for each $j \in [n]$ let $$Y_j^*= \sum_{Z \in \calZ_j \setminus \{Z^*\}}Z.$$

\noindent Define $y_j^*$ to be the value $Y_j^*$ takes for the outcome $z$, and $$v = \max_{j \in [k]} y_j^*.$$

\noindent Suppose that
$Y^*_1, \dots Y^*_m$ satisfy $y_j^*=v$ for $j \in [m]$. 
Let $\Omega$ be the set of all outcomes of $\calZ \cup \{Z_1, \dots, Z_m\}$ that agree with $z$ on the elements of $\calZ \setminus \{Z^*\}$.

\vspace{1em}
\noindent If for every index $j \in [m]$,  $Y_j$ does not depend on $Z^*$ then $$\E_{\Omega}\left[\max_{j \in [k]}Y_j'\right] = \E_{\Omega}\left[\max_{j \in [k]}Y_j\right].$$
Therefore we assume that $m' \geq 1$ of the m elements in $Y_1^*, \dots, Y_m^*$ depend on $Z^*$.

\vspace{1em}
\noindent Observe that in this case $$\E_{\Omega}\left[\max_{j \in [k]}Y_j\right] = v + p_1$$ and 
$$\E_{\Omega}\left[\max_{j \in [k]}Y_j\right] = v + (1-(1-p_1)^{m'}),$$ concluding the proof.

 \end{proof}


\vspace{1em}
\begin{lemma}\label{lem:scaled_prob_bin} Let $X_1, \dots, X_k$ be $\PB(p_1, \dots, p_n)$ random variables, and let $Y_1, \dots, Y_k$ be $\PB(q_1, \dots, q_n)$ random variables, such that there exists ground sets $\calU=\{U_1, \dots, U_\ell\}$ and $\calV=\{V_1, \dots, V_\ell\}$ each of i.i.d. Bernoulli random variables satisfying the property such that for each $j=1, \dots, k$ there exists some set $\Sigma_j \subseteq[\ell]$ such that $X_j = \sum_{i \in \Sigma_j} U_i$ and $Y_j=\sum_{i \in \Sigma_j} V_i$. Then if for some $c \in [0, 1]$, $p_{m} = c q_m$ for each $m \in [n]$ , it holds that $$\E\left[\max_{j \in [k]}Y_j\right] \geq c\E\left[\max_{j \in [k]}X_j\right].$$

\end{lemma}
\begin{proof}

    For each $\ell' \in [\ell]$ let $$U'_{\ell'} = \begin{cases}
        U_1 & U_1 = 0\\
        Z_1 & \text{otherwise}
    \end{cases}$$

    \noindent where $Z_1, \dots, Z_\ell$ are i.i.d. $\Be(c)$ random variables independent to each element of $\calU$.
    For each $j \in [k]$ define $Y_j' = \sum_{i \in \Sigma_j}U_{i}'$, and observe that $Y_j'$ and $Y_j$ are identically distributed.

\vspace{1em}
    \noindent Fix an outcome $u$ of all the variables in $\calU$ and let $\Omega$ be the set of all outcomes of $\calU \cup \{U_1', \dots, U_\ell'\}$ that agree with $u$ on $\calU$.

    \vspace{1em}
    \noindent Observe that $$\E_{\Omega}\left[\max_{j \in [k]}Y_j'\right] \geq c\E_{\Omega}\left[\max_{j \in [k]}X_j\right],$$ completing the proof.

\end{proof}

\vspace{2em}

\begin{lemma}[Balls and bins]\label{lem: lower bound on expected number of full bins }
Let \( m \) balls be thrown uniformly at random into \( n \) bins. Let $X$ be the number of bins with at least one ball. Then:
\begin{align*}
    \E [X] \geq \min \left\{ \frac{m}{2}, \frac{3n}{10}\right\} 
\end{align*}
\end{lemma}
\begin{proof}
Let $Y_i = \Ind{\text{bin $i$ is empty}}$. Then $\Pr[Y_i = 1] = (1 - \frac{1}{n})^m$. Thus, the expected number of empty bins is $\E [\sum_{i \in [n]} Y_i] = n  (1 - \frac{1}{n})^m \leq n \cdot e^{-\frac{m}{n}}$. Since $\E [X] = n -\E [\sum_{i \in [n]} Y_i] $ we get that:
\begin{align*}
     \E [X] \geq n \cdot (1 - e^{-\frac{m}{n}})
\end{align*}
Let $X_i = \Ind{\text{ball $i$ does not fall into the same bin with a different ball}} $. $\E [X_i] \geq 1 - \frac{m-1}{n}$. It is obvious that:
\begin{align*}
     \E [X] \geq \sum_{i \in [m]} \E [X_i] \geq m \left( 1 - \frac{m-1}{n}\right)
\end{align*}

We now consider two cases, namely $m  < \frac{n}{2} + 1$ and $m  \geq \frac{n}{2} + 1$. In the first case we have that:
\begin{align*}
     m \left( 1 - \frac{m-1}{n}\right) \geq \frac{m}{2}
\end{align*}

\noindent while in the second case we get that:
\begin{align*}
    n \cdot (1 - e^{-\frac{m}{n}}) &\geq n \cdot (1 - e^{-\frac{\frac{n}{2} + 1}{n}})\\
     &\geq n \cdot (1 - e^{-1})\\
     &\geq 0.3 n 
\end{align*}

\end{proof}

\newpage
\section{Hardness results}\label{sec: Hardness results}


We show that for general distributions it is NP-hard to obtain a polynomial time algorithm with an approximation ratio $(1+1/e + \epsilon)$ for any $\epsilon > 0$ via a reduction from max k-cover.
Recall that in an instance of max k-cover we are given a ground set $U=\{d_1, \dots, d_n\}$, a collection $\calS=\{S_1, \dots, S_m\}$ of subsets of $U$ and a parameter $k$. We want to select a subset $\calS'$ of $\calS$ of size $k$ that covers such that the union of sets in $\calS'$ has maximum cardinality amongst all such subsets $\calS'$.  We cite the following theorem for the hardness of max-$k$-cover.

\begin{theorem}[\hspace{-0.05em}\cite{feige1998threshold}] \label{thm:threshold} For any $\epsilon > 0$ max k-cover cannot be approximated in polynomial time within a ratio of $(1-1/e + \epsilon)$ unless P=NP.
\end{theorem}

\begin{theorem} For general distributions $\calD$ over $2^E$, it is NP-hard to approximate the portfolio problem within a ratio of $(1-1/e + \epsilon)$.
\end{theorem}
\begin{proof}
    Suppose we are given an instance $(U, \calS)$ of max k-cover, we construct an instance of portfolio optimization as follows.
    Let the ground set $E$ be equal to $\calS$, and take $M$ the 1-uniform matroid over $E$.
    Our distribution $\calD$ is the uniform distribution over $n$ outcomes so that outcome $j \in [n]$ is as follows:
    the set of active elements contains all elements of $E$ that contain the $j$-th element of $U$, that is $d_j$.
    Since the reduction approximation preserving, we obtain our result by applying Theorem \ref{thm:threshold}.
\end{proof}\label{app:hard} \newpage
\section{Contention Resolution Scheme}

In this section we present a simple contention resolution scheme (Algorithm~\ref{alg:CRS})
and prove that it satisfies \cref{thm:CRS-main}. We start by re-stating the definition
of feasible sampling strategies, that is the sampling procedure that will be used
together with our contention resolution scheme.

\deffeasiblesampling*

We are now ready to describe our simple contention resolution scheme.
Given the sampling strategy and the vector $\bp$ of marginal probabilities,
the algorithm begins by creating an ordering of the elements, by putting
elements that are less likely to be spanned by other sampled elements towards
the end of the ordering. This ordering can be calculated once and then used to
trim down any set that has been produced from the same sampling strategy.
Then, for an input set $R$ that has been sampled from the given strategy,
the algorithm first discards each element, independently, with probability $1/4$
and then runs the \textit{Greedy} algorithm, on the remaining elements, with
the order that it previously produced. The pseudocode of this procedure is given
in Algorithm~\ref{alg:CRS}.

\begin{algorithm}[H]
\caption{Contention Resolution Scheme}\label{alg:CRS}
\begin{algorithmic}
\Function{Contention-Resolution-Scheme}{$R$, $\bp$}
\State Let $R$ be sampled from a feasible sampling strategy with parameter $\bp$.
\State Let $M = (E,\calI)$ be the given matroid and $n = |E|$.
\State
\State $\text{Order} \gets []$
\State $E_r \gets E$
\For{$i = 1,\dots,n$} \Comment{Calculate the ordering of elements.}
    \State Let $x_e = p_e/4$ for $e \in E_r$ and $0$ otherwise
    \State Let $R_x$ be the output of the sampling strategy with parameter $x$.
    \State Estimate $\Pr[e \in \spann(R_x \setminus \{e\}) \big| e \in R_x]$ for all $e \in E_r$
    \State Let $e^* = \argmin_{e \in E_r} \Pr[e \in \spann(R_x \setminus \{e\}) \big| e \in R_x]$
    \State $\text{Order}[n-i-1] = e^*$
    \State $E_r \gets E_r \setminus \{e^*\}$
\EndFor
\State
\State $I \gets \emptyset$
\State $R \gets$ remove each element of $R$, independently, with probability $1/4$.
\State
\For{$i = 1,\dots,n$} \Comment{Run Greedy with the produced order.}
    \State $e \gets \text{Order}[i]$
    \If{$e \in R$ and $(I\cup\{e\})\in \calI$}
        \State $I \gets I \cup \{e\}$
    \EndIf
\EndFor
\State 
\State \Return $I$
\EndFunction
\end{algorithmic}
\end{algorithm}

\noindent We begin our analysis by proving the following useful lemma.
\begin{lemma}\label{lem:CRS-not-spanned}
Let $M=(E, \calI)$ be a matroid, $\bp \in [0,1]^{|E|}$ be a vector such that $\bp / 2 \in \calP(M)$, $\Select$ be a feasible sampling procedure, and $R$ be the output of $\Select(p)$. Let also $\tilde{R}$ be a random set that includes each element of $R$
independently with probabiliy $1/4$. Then, there exists an element $e \in E$ for which
$$
\Pr\left[e \not\in \spann(\tilde{R} \setminus \{e\}) \big| e \in \tilde{R} \right] \geq \frac{1}{2}.
$$
\end{lemma}
\newcommand{\tR}{\widetilde{R}}

\begin{proof}[Proof of Lemma~\ref{lem:CRS-not-spanned}]
Notice that since $R$ has been sampled from a feasible sampling strategy with parameter $\bp$, then $\tR$ can be seen as the sample of a feasible sampling strategy
with parameter $\bp/4$. Let $\bx \in [0,1]^{|E|}$ be a vector such that for all $e \in E$, $x_e = p_e / 2$. Then the following system of inequalities holds.

\setcounter{equation}{0}
\begin{align}
    \sum_{e \in E} x_e \cdot \Pr[e \in \spann (\tR\setminus \{e)\}\big|\ e\in \tR] 
    &\leq \sum_{e \in E} x_e \cdot \Pr[e \in \spann (\tR\setminus \{e\})]\\
    &\leq \sum_{e \in E} x_e \cdot \Pr[e \in \spann (\tR)]\\
    &= \E\left[\sum_{e \in \spann(\tR)} x_e\right]\\
    &\leq \E[|\rank(\spann(\tR))|]\\
    &= \E[|\rank(\tR)|]\\
    &\leq \E[|\tR|]\\
    &=\sum_{e \in E} p_e \cdot \frac{1}{4}\\
    &=\sum_{e \in E} x_e \cdot \frac{1}{2},
\end{align}

\noindent where we got (1) because $\tR$ has been produced by a feasible sampling strategy,
to get (2) we used the fact that $\spann(\tR \setminus \{e\}) \subseteq \spann(\tR)$,
(3) is true by the linearity of expectation, to get (4) we used that $\bx = \bp \in \calP(M)$, to get (7) we used the definition of $\tR$ and to get (8) we used the
definition of $\bx$.

\vspace{1em}
\noindent At this point the proof is concluded since the inequality between the first and the last
term tells us that there exists an element $e \in E$ such that
$$
\Pr\left[e \not\in \spann(\tilde{R} \setminus \{e\}) \big| e \in \tilde{R} \right] \geq \frac{1}{2}.
$$
\end{proof}

We are now ready to prove the main theorem of the analysis, which we restate below for convenience.

\crsmainthm*

\begin{proof}[Proof of Theorem~\ref{thm:CRS-main}]
    Fix an element $e \in E$ and suppose that it's the $j$-th element in the
    ordering produced by Algorithm~\ref{alg:CRS}. Let $O_j$ denote the first $j$
    elements of this order. The event $e \in \pi(R)$, given that $e \in R$, happens
    if and only if (1) $e$ is not discarded during the downsampling, that is $e \in \tR$
    and (2) $e$ is picked by the Greedy algorithm, that is $e \not \in \spann(\tR \cap O_j \setminus \{e\}).$

    \vspace{1em}
    \noindent Due to the construction of the algorithm's ordering, we know that for any
    $R'$ produced by a feasible sampling strategy with parameter $\bx$, where $x_e = p_e /4$ for $e \in O_j$ and $x_e = 0$ otherwise, it holds that
    $$
    e = \argmax_{e' \in O_j} \Pr \left[ e' \not\in \spann(R' \setminus \{e'\}) \big | e' \in R'\right]
    $$

    \noindent Due to Lemma~\ref{lem:CRS-not-spanned}, we know that there exists
    an element $e' \in O_j$ such that $$\Pr \left[ e' \not\in \spann(R' \setminus \{e'\}) \big | e' \in R'\right] \geq 1/2.$$ 
    \noindent Since $e'$ has the biggest such probability, it's true that
    $$
    \Pr \left[ e \not\in \spann(R' \setminus \{e\}) \big | e \in R'\right] \geq \frac{1}{2}.
    $$

    \noindent In addition, $R'$ can be seen as first sampling the set $\tR$
    and then discarding all elements in $\tR \setminus O_j$. Therefore, the above 
    inequality becomes
    
    $$
    \Pr \left[ e \not\in \spann(\tR \cap O_j \setminus \{e\}) \big | e \in \tR\right] \geq \frac{1}{2}.
    $$
    
    \noindent The proof of the theorem is concluded as follows.
    \begin{align*}
        \Pr \left[ e \in \pi(R) \big | e \in R \right] &= \Pr \left[ e\not \in \spann(\tR \cap O_j \setminus \{e\}) \wedge e \in \tR\ \big | e \in R \right]\\
        &= \Pr \left[ e\not \in \spann(\tR \cap O_j \setminus \{e\}) \big | e \in \tR \right] \cdot \Pr \left[ e \in \tR\ \big | e \in R \right]\\
        &\geq \frac{1}{2} \cdot \frac{1}{4} = \frac{1}{8}
    \end{align*}
\end{proof}

\label{app:CR}\newpage
\section{Deferred Proofs from Section 3}\label{sec:appendix-bruteforce-construction}

In this section, we discuss the counter-example for taking disjoint solutions that is mentioned in Section~\ref{sec: Our Techniques}. Recall that the instance is the following: we have $k\cdot \log k$ elements
with activation probability $\nicefrac{1}{k}$ and the rest of the elements have activation probability $\nicefrac{1}{k^2}$. Also, for convenience we pick the rank $r$ of the uniform matroid to be equal to $k$.
The disjoint algorithm will form $\log k$ disjoint solutions with expectation $1$ and will then fill its
portfolio with disjoint solutions formed by the ``low'' probability elements. This portfolio will 
achieve a value of $O(\log \log k)$. Surprisingly, if one completes their portfolio by picking more subsets
from the ``high'' probability items instead of using the ``low'' probability ones, they can achieve a value that is near-exponentially better (and also asymptotically optimal).

\vspace{1em}
\noindent To ease the notation, let $B = \nicefrac{1}{2} \cdot \nicefrac{\log k}{\log \log k}$.
The construction is the following. We will use the first $B \cdot k$ elements and arrange them
into $B$ rows of $k$ elements each. Remember that $r = k$, so every $k$ elements form a feasible solution.
We will group the elements by \textbf{creating batches} that are of size $k/B$. This means that each row will
be split up into $B$ batches:
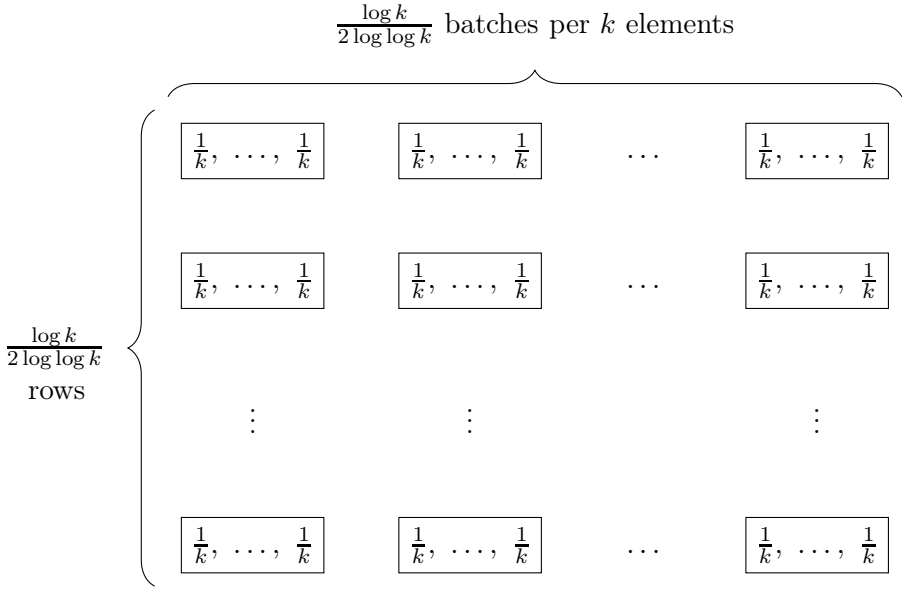
\usetikzlibrary{matrix,decorations.pathreplacing}

\[
\hspace{-2cm} 
\begin{tikzpicture}
    \matrix (m) [matrix of math nodes, 
                 nodes in empty cells,
                 row sep=1cm, 
                 column sep=1cm] {
        \frac{1}{k},\ \dots,\ \frac{1}{k} & \frac{1}{k},\ \dots,\ \frac{1}{k} & \dots & \frac{1}{k},\ \dots,\ \frac{1}{k} \\
        \frac{1}{k},\ \dots,\ \frac{1}{k} & \frac{1}{k},\ \dots,\ \frac{1}{k} & \dots & \frac{1}{k},\ \dots,\ \frac{1}{k} \\
        \vdots                           & \vdots                           &       & \vdots                           \\
        \frac{1}{k},\ \dots,\ \frac{1}{k} & \frac{1}{k},\ \dots,\ \frac{1}{k} & \dots & \frac{1}{k},\ \dots,\ \frac{1}{k} \\
    };

    \draw[draw] (m-1-1.north west) rectangle (m-1-1.south east);
    \draw[draw] (m-1-2.north west) rectangle (m-1-2.south east);
    \draw[draw] (m-1-4.north west) rectangle (m-1-4.south east);

    \draw[draw] (m-2-1.north west) rectangle (m-2-1.south east);
    \draw[draw] (m-2-2.north west) rectangle (m-2-2.south east);
    \draw[draw] (m-2-4.north west) rectangle (m-2-4.south east);

    \draw[draw] (m-4-1.north west) rectangle (m-4-1.south east);
    \draw[draw] (m-4-2.north west) rectangle (m-4-2.south east);
    \draw[draw] (m-4-4.north west) rectangle (m-4-4.south east);

    \draw[decorate,decoration={brace,amplitude=10pt}] 
        ([xshift=-10pt,yshift=-5pt]m-4-1.south west) -- ([xshift=-10pt,yshift=5pt]m-1-1.north west)
        node[midway,xshift=-15pt,anchor=east] {$\frac{\log k}{2\log \log k}$}
        node[midway,xshift=-22pt,yshift=-17pt,anchor=east] {rows};

    \draw[decorate,decoration={brace,amplitude=10pt}] 
        ([yshift=10pt,xshift=-5pt]m-1-1.north west) -- ([yshift=10pt,xshift=5pt]m-1-4.north east)
        node[midway,yshift=15pt,anchor=south] {$\frac{\log k}{2\log \log k}$ batches per $k$ elements};
\end{tikzpicture}
\]

\vspace{1em}
\noindent Notice that any selection of $B$ batches forms a feasible solution. We will simply construct our portfolio by including all possible combinations of $B$ batches. This is feasible because
$$
\binom{B^2}{B} \leq e^B \cdot B^B \leq k.
$$

\noindent To analyze the value of the constructed portfolio, notice that the expected number of active elements in the above matrix is $k \cdot B \cdot \nicefrac{1}{k} = B$. Since this random variable follows
a binomial distribution, we know that with constant probability, say $1/2$, there will be $\Omega(B)$
active elements in the above matrix.

\vspace{1em}
\noindent Fix any instance where there are at least $B$ active elements in the above matrix. Sort the batches
in decreasing order of their number of active elements. The first $B$ batches must contain at least $B$
active elements in total. Since our portfolio contains these $B$ batches together in one solution,
the portfolio's value for this instance will be at least $B = \log k / 2 \log \log k$, which concludes the proof.

\vspace{1em}
\noindent One can also show that this construction is within a constant factor of $\OPT$, since
the solutions picked by $\OPT$ will be binomials with expectation $1$. Since those will be positevely dependent, the expectation of their maximum will be at most the expectation of the maximum of $k$ independent
such binomials, which is $\Theta(\log k / \log \log k)$.

\end{document}